\documentclass{amsart}
\pdfoutput=1

\usepackage{amssymb}
\usepackage{amsthm}
\usepackage{amsfonts}
\usepackage{amsmath}
\usepackage{pmboxdraw}
\usepackage{verbatim}
\usepackage{graphicx}
\usepackage{color}
\usepackage[colorlinks=true, citecolor=blue, filecolor=black, linkcolor=black, urlcolor=black]{hyperref}
\usepackage{cite}
\usepackage[normalem]{ulem}
\usepackage{subcaption}
\usepackage{bbm}
\usepackage{bm}
\usepackage{mathtools}
\usepackage{todonotes}
\usepackage{kantlipsum}
\usepackage{booktabs}
\usepackage{multirow}
\usepackage{array}
\usepackage{nicematrix}
\usepackage{float}
\usepackage[utf8]{inputenc}
\usepackage{marginnote}

\usepackage{tcolorbox}
\tcbuselibrary{most}

\allowdisplaybreaks

\newtcolorbox{defn}[1]{breakable,
colbacktitle=gray!50!white, fonttitle=\bfseries, coltitle=black, title=Definition: {#1}}

\usepackage{tikz}
\usetikzlibrary{positioning}

\newcommand{\RN}[1]{%
	\textup{\uppercase\expandafter{\romannumeral#1}}%
}

\def\bfs{\boldsymbol}

\def\wt{\widetilde}

\def\Mot{ \mathrm{Mot}}

\def\C{\mathbb{C}}

\def\R{\mathbb{R}}

\def\wt{ \mathrm{wt}}

\newcommand{\im}{\operatorname{Im}}

\newcommand{\supp}{\operatorname{supp}}

\usepackage{tikz-cd} 
\newcommand{\qbinom}[2]{{\begin{bmatrix}#1\\#2\end{bmatrix}}_q}
\newcommand{\floor}[1]{\left\lfloor#1\right\rfloor}
\newcommand{\set}[1]{\left\{#1\right\}}

\theoremstyle{plain}
\newtheorem{thm}{Theorem}[section]

\newtheorem{lem}[thm]{Lemma}

\newtheorem{prop}[thm]{Proposition}

\theoremstyle{remark}
\newtheorem{rem}{Remark}

\newcommand{\abs}[1]{\lvert#1\rvert}
\numberwithin{equation}{section}

\begin{document}

\title[Spectral analysis of $q$-deformed Al-Salam--Carlitz unitary ensembles]{Spectral analysis of $q$-deformed unitary ensembles \\ with the Al-Salam--Carlitz weight}
\author{Sung-Soo Byun}
\address{Department of Mathematical Sciences and Research Institute of Mathematics, Seoul National University, Seoul 151-747, Republic of Korea}
\email{sungsoobyun@snu.ac.kr}

\author{Yeong-Gwang Jung}
\address{Department of Mathematical Sciences, Seoul National University, Seoul 151-747, Republic of Korea}
\email{wollow21@snu.ac.kr} 

\author{Jaeseong Oh}
\address{June E Huh Center for Mathematical Challenges, Korea Institute for Advanced Study, 85 Hoegiro, Dongdaemun-gu, Seoul 02455, Republic of Korea}
\email{jsoh@kias.re.kr}


\begin{abstract} We study $q$-deformed random unitary ensembles associated with the weight function of the Al-Salam--Carlitz orthogonal polynomials, indexed by a parameter $a < 0$. In the special case $a = -1$, the model reduces to the $q$-deformed Gaussian unitary ensemble. Employing the Flajolet--Viennot theory together with the combinatorics of matchings, we derive an explicit positive-sum expression for the spectral moments. In the double-scaling regime $q = e^{-\lambda/N}$, where $N$ denotes the ensemble size and $\lambda > 0$ is fixed, we derive the first two terms in the large-$N$ expansion of the spectral moments. As a consequence, we obtain a closed-form expression for the limiting spectral density. Notably, this density exhibits two successive phase transitions as $\lambda$ increases, characterised by a reduction in the number of soft edges from two, to one, and eventually to none. Furthermore, we show that the limiting density coincides with the limiting zero distribution of the Al-Salam--Carlitz orthogonal polynomials under the same scaling.
\end{abstract}

\maketitle 

\section{Introduction and results}
The concept of \( q \)-deformation, where \( q \in (0,1) \), offers a powerful framework for constructing quantised analogues of classical mathematical structures. A fundamental example is the \( q \)-integer,
\begin{equation} 
[n]_q = \frac{1 - q^n}{1 - q},
\end{equation} 
which recovers the standard integer \( n \) in the limit \( q \to 1 \). Such deformations play a central role in the theory of orthogonal polynomials--most notably in the \( q \)-Askey scheme \cite{As80, KLS10,ISV87}--and are deeply connected with exactly solvable models in mathematical physics. In particular, they lead to rich families of determinantal point processes that serve as \( q \)-analogues of classical random matrix ensembles. The spectral properties of these models reflect \( q \)-dependent structures and reveal phenomena absent in the classical continuum limit \( (q = 1) \).

In this work, we investigate a natural $q$-deformation of a classical random matrix ensemble \cite{Fo10}. To introduce this model, we begin by defining the associated weight function.
Given parameters $a < 0$ and $q \in (0,1)$, we consider the weight function
 \begin{equation} \label{def of AlSalam weight}
\omega_U^{(a)}(x) \equiv \omega_U^{(a)}(x;q) := \frac{ (qx,qx/a;q)_\infty }{ (q,a,q/a;q)_\infty}, \qquad (a<x<1),  
 \end{equation}
where 
\begin{align} 
(a_1,\dots, a_k ; q)_\infty = \prod_{l=0}^\infty (1-a_1 q^l) \dots (1-a_k q^l).
\end{align}
The weight function \eqref{def of AlSalam weight} defines the orthogonality measure for the Al-Salam--Carlitz polynomials, see Subsection~\ref{Subsection_Al-Salam poly} for details.
For the special case $a=-1$, the weight \eqref{def of AlSalam weight} corresponds to that of the discrete $q$-Hermite polynomials.  
It is often convenient, particularly when analysing the continuum limit \( q \to 1 \), to adopt the following scaling:
\begin{equation} \label{def of AlSalam weight hat}
\widehat{\omega}_U^{ (a) } (x) \equiv \widehat{\omega}_U^{ (a) } (x;q) :=  \sqrt{1-q}\,\omega_U^{  (a) } ( \sqrt{1-q}\,x ). 
\end{equation}
Then we have 
\begin{equation} \label{limit of weight q to 1}
\lim_{ q \to 1^- }   \widehat{w}_U^{(a)}(x;q) \Big|_{a=-q^r} = \frac{1}{\sqrt{2\pi}} e^{ -(x-r)^2/2 },
\end{equation}
thereby recovering the (shifted) Gaussian weight in the continuum limit, see e.g. \cite[Eq.~(2.11)]{BF00}. 

We consider the system $\bfs{x}=\{x_j\}_{j=1}^N$ supported on the $q$-lattice $\{1,q,q^2,\dots \} \cup \{a,aq,aq^2,\dots \}$ having the joint distribution of the form (see e.g. \cite{Ol21,BF25a})
\begin{equation} \label{def of jpdf}
\frac{1}{Z_N}\prod_{1\le j<k \le N}(x_j-x_k)^2 \prod_{j=1}^N \omega_U^{(a)}(x_j), 
\end{equation}
where $Z_N$ is the normalisation constant. We refer this ensemble as the \textit{$q$-Al-Salam--Carlitz unitary ensemble}.
The special case $a = -1$, known as the \textit{$q$-deformed Gaussian unitary ensemble} (GUE), has been extensively studied in recent literature  \cite{BFO24,FLSY23,Co21,CCO20, Le05, LSYF23,MPS20} as a quantised version of the classical GUE. 
In this work, we extend previous results to the more general ensemble \eqref{def of jpdf} with arbitrary \( a < 0 \). Unlike the special case \( a = -1 \), where underlying symmetries simplify the analysis, the general case introduces new structural features. Before providing the details, we briefly summarise our main results.

\begin{itemize}
    \item \textbf{(Closed form of spectral moments; Theorem~\ref{Thm_spectral moments})} We derive closed-form expressions for the spectral moments of the ensemble using techniques from algebraic combinatorics, particularly the Flajolet--Viennot theory.
    \smallskip
    \item \textbf{(Large-$N$ expansion of spectral moments; Theorem~\ref{Thm_genus expansion})} Based on the positive-sum formula in Theorem~\ref{Thm_spectral moments}, we establish the large-$N$ expansion of the spectral moments. 
    \smallskip
    \item \textbf{(Limiting spectral density; Theorem~\ref{Thm_limiting density})} We determine the limiting spectral density on the exponential lattice. The result reveals a phase transition, with the density exhibiting three distinct regimes---by contrast with the two distinct regimes observed in the \( q \)-deformed GUE \cite{BFO24}. Moreover, we verify that the limiting spectral density agrees with the asymptotic zero distribution of the Al-Salam--Carlitz polynomials (Proposition~\ref{prop:Asymptotic zero distribution}).
\end{itemize}

Let us be more precise in introducing our main results. In the spectral analysis of random matrix ensembles, the fundamental observable is the \textit{spectral moments} defined by 
\begin{equation}
m_{N,p}^{ (a) }:= \mathbb{E} \Big[ \sum_{j=1}^N x_j^{p} \Big],
\end{equation}
where the expectation is taken with respect to \eqref{def of jpdf}. The importance of spectral moments in random matrix theory dates back to the foundational work of Wigner \cite{Wi55}, where the spectral moments of the GUE were analysed via moment-counting methods, leading to the celebrated semicircle law. Since then, there has been extensive research on spectral moments in various exactly solvable models. We refer the reader to the recent, though not complete, list of related works \cite{ABGS21a,CDO21,GGR21,CMOS19,RF21,BF24,By24b,ABO25}, as well as the references therein.

To analyse the system \eqref{def of jpdf}, an effective approach is to make use of the family of \( q \)-orthogonal polynomials associated with the weight \eqref{def of AlSalam weight}. These are the Al-Salam--Carlitz polynomials $U_n^{(a)}(x) \equiv U_n^{(a)}(x;q)$, defined by the recurrence relation
\begin{equation} \label{def of three term AlSalam}
x\,U_n^{(a)}(x) = U_{n+1}^{(a)}(x) +(a+1)q^n U_n^{(a)}(x) -a q^{n-1} (1-q^n)U_{n-1}^{(a)}(x)
\end{equation}
for $n \ge 1$, with initial conditions $U_{-1}^{(a)}(x)=0$ and $U_0^{(a)}(x)=1$. (From \eqref{def of three term AlSalam}, one can observe a notable simplification when \( a = -1 \), in which case the second term on the right-hand side vanishes.)
Using the theory of determinantal point processes, the spectral moments can be expressed in terms of the averaged density
\begin{equation} \label{def of 1pt density}
\rho_N^{(a)}(x) :=  \sum_{ j=0 }^{N-1} \frac{ U_{j}^{(a)}(x;q)^2  }{ (-a)^{j}(1-q)(q;q)_{j} q^{\frac{j(j-1)}{2}} }  \,\omega_U^{(a)}(x), 
\end{equation}
where  
$(a;q)_n  := (1-a)(1-aq) \dots (1-aq^{n-1})$. It then follows that 
\begin{equation} \label{def of moment in terms of q-integral}
m_{N,p}^{ (a) } =  \int_a^1 x^{p} \rho_N^{(a)}(x)\,d_q x, 
\end{equation}
where $d_q x$ denotes the Jackson $q$-integral, see Subsection~\ref{Subsection_Al-Salam poly}.  

In the analysis of spectral moments, one often aims to derive a closed-form expression to facilitate asymptotic analysis. For instance, explicit formulas \eqref{def of GUE moments} for the spectral moments of the GUE can be found in \cite{HT03,Ke99}, from which one can recover not only Wigner’s semicircle law \cite{Wi55}, but also the well-known Harer--Zagier recursion formula \cite{HZ86}. (See also \cite{MS11,MS12,MS13,CMSV16} for further applications.) However, such derivations are usually significantly more challenging for \( q \)-deformed models, as the relevant Jackson $q$-integrals can be evaluated explicitly only in special cases. In our first main result, we obtain a closed-form expression for \( m_{N,p}^{(a)} \) in the general case.
To describe it, we first introduce the following notation: for $b,c \in \mathbb{N}$, 
    \begin{equation}\label{def of H}
        \mathsf{H}(b,c) :=\sum_{0\leq j_{1}\leq j_{2}\leq\cdots\leq j_{c}\leq b}\prod_{k=1}^{c}\frac{[2j_{k}+k-2]_{q}!!}{[2j_{k}+k-1]_{q}!!}.
    \end{equation} 
Here, in the case $b=0$ or $c=0$, $\mathsf{H}(b,c)$ is defined by
\begin{equation}\label{def of H special cases}
\mathsf{H}(b,0)=1, \qquad \mathsf{H}(0,c)=\frac{1}{[c-1]_{q}!!}.
\end{equation}
Then we have the following.

\begin{thm}[\textbf{Spectral moments of Al-Salam--Carlitz unitary ensemble}] \label{Thm_spectral moments}
For any $p, N \in \mathbb{N}$ and $a<0$, we have 
\begin{align}
\begin{split} \label{moment closed main}
m_{N,p}^{ (a) }  &=  \sum_{j=0}^{N-1}\sum_{k=0}^{\floor{p/2}} \frac{    (-a)^{k} (1-q)^k}{ (a+1)^{2k-p} } \sum_{l=0}^k     q^{-l(p-l)+\frac{l(l-1)}{2}}\frac{[p]_{q}!}{[p-2l]_{q}!! \, [l]_{q}!}\mathsf{H}(k-l,p-2k)q^{j(p-l)}\qbinom{j}{l}.
\end{split}
\end{align}
Furthermore, we have  
\begin{equation}\label{eqn:symmetry of spectral moment}
m_{N,p}^{(1/a)} = \frac{1}{a^p} m_{N,p}^{(a)}. 
\end{equation}
\end{thm}

\begin{rem}[Moments of the Al-Salam--Carlitz polynomial]
Note that by \eqref{def of 1pt density} and \eqref{def of moment in terms of q-integral}, Theorem~\ref{Thm_spectral moments} yields  
\begin{align}
\begin{split}
 \int_a^1  x^p \, U_{j}^{(a)}(x;q)^2   \,\omega_U^{(a)}(x) \,d_qx 
&= (-a)^{j}(1-q)(q;q)_{j} q^{\frac{j(j-1)}{2}} \sum_{k=0}^{\floor{p/2}}  \frac{    (-a)^{k} (1-q)^k}{ (a+1)^{2k-p} }
\\
&\quad \times   \sum_{l=0}^k     q^{-l(p-l)+\frac{l(l-1)}{2}}\frac{[p]_{q}!}{[p-2l]_{q}!! \, [l]_{q}!}\mathsf{H}(k-l,p-2k)q^{j(p-l)}\qbinom{j}{l}. 
\end{split}
\end{align}
In particular, for the special case $j=0$, it gives rise to  
\begin{align}
\begin{split}
 \int_a^1  x^p   \,\omega_U^{(a)}(x) \,d_qx 
&= (1-q)  \sum_{k=0}^{\floor{p/2}}  \frac{    (-a)^{k} (1-q)^k}{ (a+1)^{2k-p} }  \frac{[p]_{q}!}{[p]_{q}!! }\mathsf{H}(k,p-2k). 
\end{split}
\end{align}
Such a moment has been studied in the literature both from a combinatorial perspective \cite{Kim97} and as the moment of a certain random variable \cite{AY25}.
Furthermore, for the special case $a=-1$, we have 
\begin{equation}
 \int_{-1}^1  x^{2p}   \,\omega_U^{(-1)}(x) \,d_qx =  (1-q)^{p+1} [2p-1]_{q}!!,
\end{equation}
which serves as a \( q \)-analogue of the classical Gaussian integral, see also \cite[Eq.~(1.41)]{BFO24}.  
\end{rem}

\begin{rem}[Comparison with the formulas form Jack/Macdonald polynomial theory]
In the recent work \cite{BF25a}, a $(q,t)$-extension of the Gaussian $\beta$-ensemble was investigated (see \cite{Le09,WF14} for the continuum case). In particular, the spectral moments were computed using symmetric function theory, specifically through the framework of Macdonald polynomials. This approach yields an alternative representation in the form of an alternating series. For the first few values, we are able to perform a consistency check with our results. To be precise, the first few spectral moments in Theorem~\ref{Thm_spectral moments} are given as follows:
\begin{equation*}
m_{N,0}^{(a)}=N, \qquad m_{N,1}^{(a)}=(a+1)\frac{1-q^{N}}{1-q},
    \qquad 
    m_{N,2}^{(a)}=\frac{1-q^{N}}{q(1-q^{2})}\Big((a^{2}+1)q+q^{N}(q+a(1+2q+q^2+aq))\Big). 
\end{equation*} 
These are consistent with \cite[Eqs.~(4.29)--(4.31)]{BF25a}, upon setting $t = q$ and $u = q^N$.
\end{rem}

Next, we investigate the large-\( N \) asymptotic expansion of the spectral moments. For the GUE spectral moments
\begin{equation}  \label{def of GUE moments}
 m_{N,2p}^{ \rm (GUE) }:= (2p-1)!! \sum_{ l=0 }^p \binom{N}{l+1}\binom{p}{l} 2^l ,
\end{equation}
the expansion occurs in powers of \( 1/N^2 \), and takes the form
\begin{equation} 
m_{N,2p}^{ \rm (GUE) }  =   C_p N^{p+1} + \sum_{g=1}^{\lfloor (p+1)/2 \rfloor}  \mathcal{E}_g(p)  N^{p+1-2g}, 
\end{equation}
where \( C_p := \frac{1}{p+1} \binom{2p}{p} \) is the Catalan number. The coefficients \( \mathcal{E}_g(p) \) can be expressed in terms of the Stirling numbers of the first kind, see \cite[Eq.~(1.27)]{BFO24}. Note that the Catalan number $C_p$ coincides with the $2p$-th moment of the semicircle law, that is,
\begin{equation} \label{def of semi-circle}
C_p = \int_{-2}^2 x^{2p} \rho_{\mathrm{sc}}(x)\,dx, \qquad \rho_{\mathrm{sc}}(x) := \frac{\sqrt{4 - x^2}}{2\pi} \mathbbm{1}_{[-2,2]}(x).
\end{equation}
This identity plays a fundamental role in the derivation of the semicircle law~\cite{Wi55}.

To analyse the large-$N$ behaviour of the $q$-deformed model, we properly scale the parameter $q$. In particular, obtaining a non-trivial limiting behaviour requires a double scaling limit in which $q \to 1$ at an appropriate rate as $N \to \infty$. Following the scaling regimes used in the Stieltjes--Wigert ensemble~\cite{Fo21} and the $q$-deformed Gaussian ensemble~\cite{BFO24}, we consider the following scaling:
\begin{equation} \label{def of q scaling}
    q=e^{-\frac{\lambda}{N}}, \qquad \lambda\geq0.
\end{equation} 
Recall also that the regularised incomplete beta function is given by 
\begin{equation} \label{def of beta ftn}
I_x(a,b)=  \frac{ \Gamma(a+b) }{ \Gamma(a)\Gamma(b) } \int_0^x t^{a-1} (1-t)^{b-1}\,dt.
\end{equation}
Then we have the following. 
 
\begin{thm}[\textbf{Large-$N$ expansion of spectral moments}] \label{Thm_genus expansion}
    Let $a < 0$ be fixed, and let $q$ be scaled according to \eqref{def of q scaling}.
 Then as $N\rightarrow\infty$, we have
    \begin{equation}\label{eqn:genus expansion}
      q^{ p/2} m_{N,p}^{ (a) }=\mathcal{M}_{p,0}N+\frac{\mathcal{M}_{p,1}}{N}+O\big(\frac{1}{N^3}\big),
    \end{equation}
    where  
    \begin{align}
    \begin{split}
     \label{def of mathcal Mp0}
        \mathcal{M}_{p,0}
       & =\frac{1}{\lambda}\sum_{l=0}^{\floor{p/2}}(a+1)^{p-2l}(-a)^{l}\frac{(p-l-1)!}{l!(p-2l)!}I_{1-\mathsf{s} }(l+1,p-l),
    \end{split}
    \\
     \begin{split}\label{def of mathcal Mp1}
    \mathcal{M}_{p,1} &=-\frac{\lambda p}{12}
        \sum_{l=0}^{\floor{p/2}}\frac{(a+1)^{p-2l}(-a)^{l}}{(p-2l)! \, l!}
        \\
        &\quad \times \bigg( {\frac{p}{2}(p-l-1)!}I_{1-\mathsf{s}  }(l+1,p-l)
        +\frac{(p-1)!}{(l-1)!}  \mathsf{s}^{p-l} (1-\mathsf{s}  )^{l-1}(p-l+2-(p+1) \mathsf{s} )\bigg).     
    \end{split}
    \end{align}
Here, $\mathsf{s}=e^{-\lambda}$. 
\end{thm}

\begin{rem}[Alternative representations]
    From \cite[Eq.~(8.17.5)]{NIST} we have
    \begin{equation} \label{def of Mp0 v2}
        \mathcal{M}_{p,0}
        =\frac{1}{\lambda}\sum_{l=0}^{\floor{p/2}}(a+1)^{p-2l}(-a)^{l}\frac{(p-l-1)!}{l!(p-2l)!}\sum_{j=l+1}^{p}\binom{p}{j}(1- \mathsf{s}  )^{j} \mathsf{s}^{p-j}. 
    \end{equation}
    In addition, by using \eqref{def of beta ftn}, we have 
    \begin{equation} \label{def of Mp0 v3}
       \mathcal{M}_{p,0}= \frac{1}{\lambda}\int_{0}^{1-\mathsf{s} } (a+1)^{p}(1-t)^{p-1} {{}_2F_{1}}\Big(\frac{1-p}{2},-\frac{p}{2}\, ; \, 1 \, ; \, \frac{-4at}{(1-t)(a+1)^2}\Big)\, dt, 
    \end{equation}
    where ${}_2F_1$ is the hypergeometric function \cite[Chapter 15]{NIST}. 
\end{rem}

In our next result, we derive the limiting spectral density. 

\begin{thm}[\textbf{Limiting spectral density}] \label{Thm_limiting density}
Let $a < 0$ be fixed, and let $q$ be scaled according to \eqref{def of q scaling}.
As $N \to \infty$, in the sense of integration against continuous test functions $f \in C([a,1])$, we have 
\begin{equation}
\frac{1}{N}\rho^{ (a) }_N( x ) \to  \rho^{ (a) }( x ), \qquad (a<x<1).
\end{equation}
Here, $\rho^{(a)}$ satisfies the symmetry
\begin{equation} \label{symmetry of rho a}
\rho^{(1/a)}(x) = -\frac{1}{a}\, \rho^{(a)}\Big( \frac{x}{a} \Big).
\end{equation}
For $a \in [-1,0)$, the density $\rho^{(a)}$ admits the following explicit expression. (The case $a \le -1$ then follows by applying the symmetry relation~\eqref{symmetry of rho a}.)
\begin{itemize}
    \item \textup{\textbf{(Support of $\rho^{ (a) }$)}} 
    Let 
     \begin{equation} \label{def of u v(lambda) main}
    \mathsf{u} \equiv    \mathsf{u}(\lambda):= (1+a)e^{-\lambda} , \qquad \mathsf{v} \equiv  \mathsf{v}(\lambda):= 2\sqrt{-a(1-e^{-\lambda})e^{-\lambda}} .
    \end{equation} 
    Then we have    
    \begin{equation}\label{def of support}
    \supp( \rho^{(a)}(x) )= \begin{cases}
    (\mathsf{u}-\mathsf{v},\mathsf{u}+\mathsf{v} ) & \textup{if } \lambda \in (0, \log(1-a) ) , 
    \smallskip 
    \\ 
    (\mathsf{u}-\mathsf{v},1 ) &\textup{if }\lambda \in (\log(1-a),\log(1-a)-\log(-a) ) ,
    \smallskip 
    \\
    (a,1) & \textup{if }\lambda \in (\log(1-a)-\log(-a),\infty) .
    \end{cases}
    \end{equation}
     \item \textup{\textbf{(Density of $\rho^{ (a) }$)}} Let 
\begin{equation}\label{eqn:def of x0x1}
        x_{0}=\frac{a^{2}+1-x(a+1)}{(a-1)^{2}},\qquad x_{1}=\frac{\sqrt{4a(x-a)(x-1)}}{(a-1)^{2}} . 
    \end{equation}
   Then we have 
       \begin{align} \label{def of limiting density}
    \begin{split}
    \rho^{(a)}(x) & = \frac{2}{\pi\lambda\abs{x} }\arctan{\sqrt{\frac{1-x_0-x_1}{1-x_0+x_1}\frac{1-e^{-\lambda}-x_0+x_1}{x_0+x_1-1+e^{-\lambda}}}}\, \mathbbm{1}_{(\mathsf{u}-\mathsf{v},\mathsf{u}+\mathsf{v} ) }( x )  
    \\
    & \quad + \begin{cases}
    0 &\textup{if } \lambda \in (0, \log(1-a) ) , 
    \smallskip 
    \\
    \displaystyle 
    \frac{1}{\lambda |x|} \, \mathbbm{1}_{(\mathsf{u}+\mathsf{v},1)} (x)
    &\textup{if }\lambda \in ( \log(1-a), \log(1-a)-\log(-a)) , 
    \smallskip 
    \\
   \displaystyle \frac{1}{\lambda |x|} \, \mathbbm{1}_{(a,\mathsf{u}-\mathsf{v}) \cup (\mathsf{u}+\mathsf{v},1) } (x) 
    &\textup{if }\lambda \in (\log(1-a)-\log(-a),\infty) . 
    \end{cases}
    \end{split}
    \end{align} 
\end{itemize} 
\end{thm}

See Figure~\ref{Fig_limiting density} for the graphs of the limiting density $\rho^{(a)}$. 
Perhaps the most remarkable feature of the limiting spectral density is the presence of a phase transition. 
\begin{itemize}
    \item For $\lambda < \log(1 - a)$, the density exhibits square-root decay at both spectral edges $\mathsf{u} \pm \mathsf{v}$, reminiscent of the semicircle law \eqref{def of semi-circle}. See Figure~\ref{Fig_limiting density} (A). 
    \smallskip 
    \item In the intermediate regime $\lambda \in (\log(1 - a), \log(1 - a) - \log(-a))$, the right endpoint of the spectrum reaches the hard edge at $x = 1$, resulting in one soft edge and one hard edge. Notably, unlike the classical hard-edge behaviour in continuous random matrix ensembles, the density remains bounded at the hard edge. See Figure~\ref{Fig_limiting density} (B).
    \smallskip 
    \item For $\lambda > \log(1 - a) - \log(-a)$, the left endpoint also reaches the hard edge at $x = a$, so that both edges become hard. See Figure~\ref{Fig_limiting density} (C).
\end{itemize}
Another remarkable feature in the presence of the hard edge is the emergence of a pronounced peak in the bulk of the spectrum—similar to what has been observed in the Stieltjes--Wigert ensemble~\cite{Fo21} and the $q$-deformed Gaussian ensemble~\cite{BFO24}.

\begin{figure}[t]
	\begin{subfigure}{0.3\textwidth}
		\begin{center}	
		 \includegraphics[width=\textwidth]{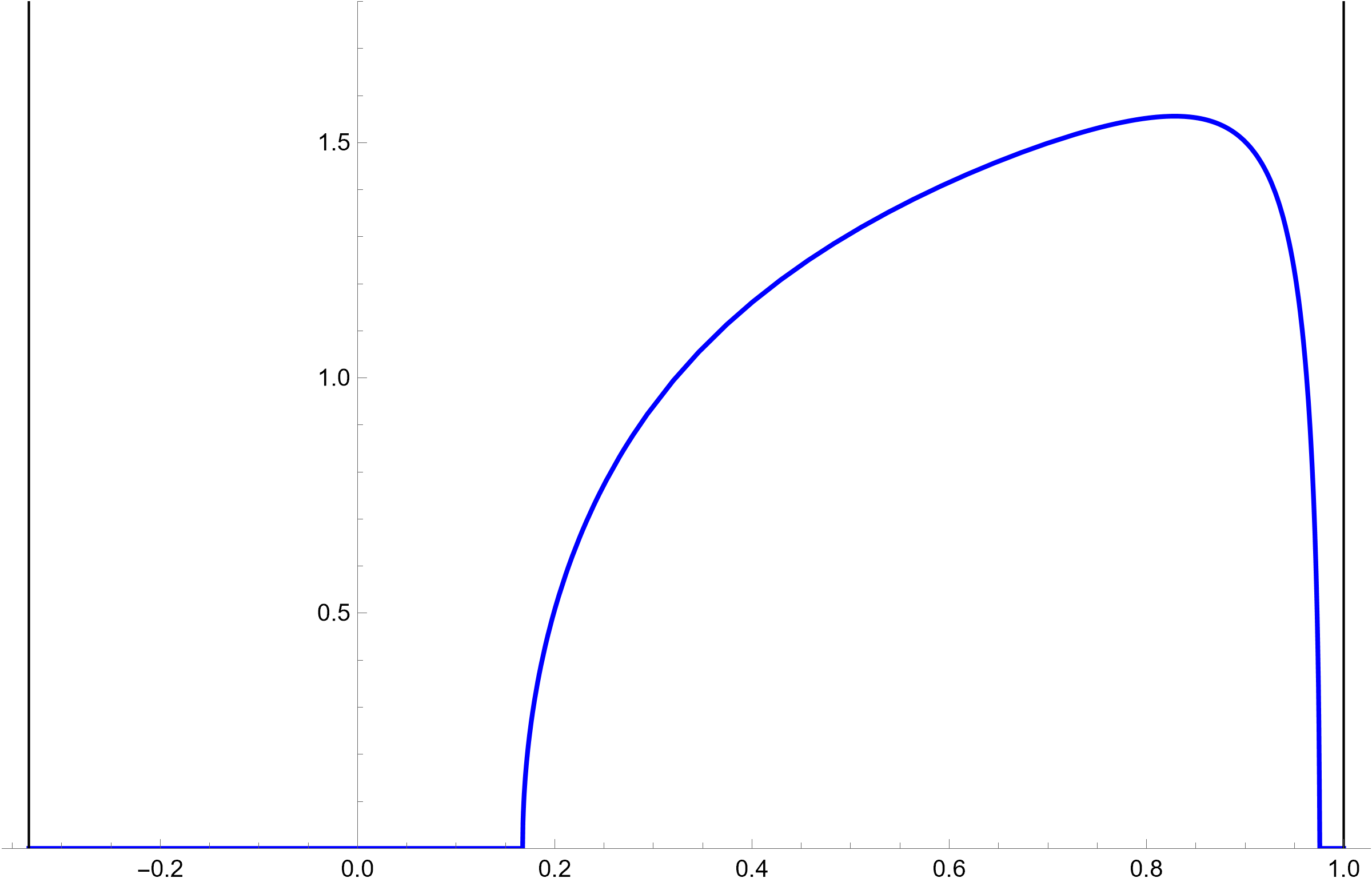}
		\end{center}
		\subcaption{$\log(7/6)$}
	\end{subfigure}	
\quad
	  	\begin{subfigure}{0.3\textwidth}
		\begin{center}	
		    \includegraphics[width=\textwidth]{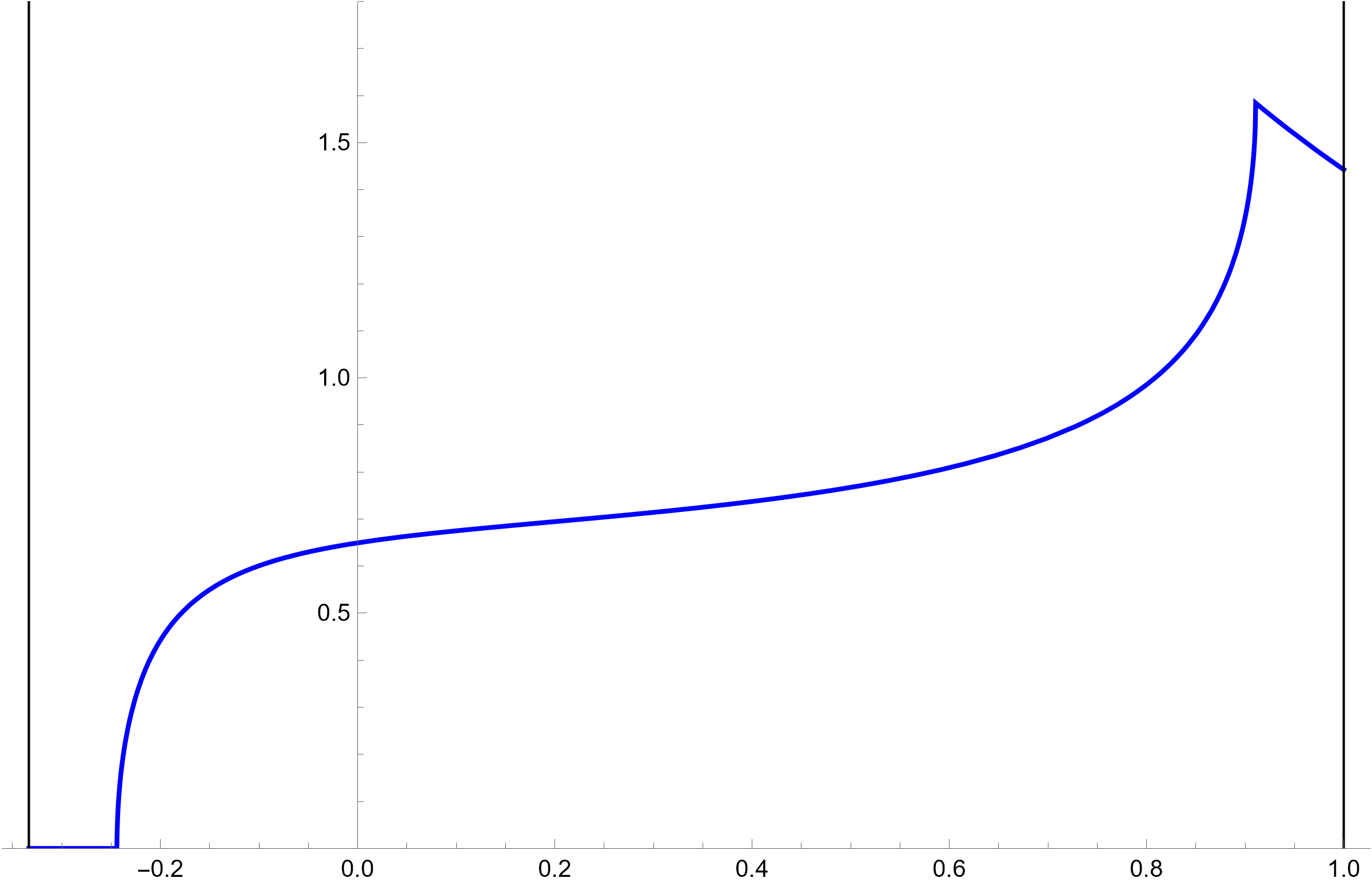}
		\end{center}
		\subcaption{$\log(2)$}
	\end{subfigure}	
 \quad
 \begin{subfigure}{0.3\textwidth}
		\begin{center}	
	   \includegraphics[width=\textwidth]{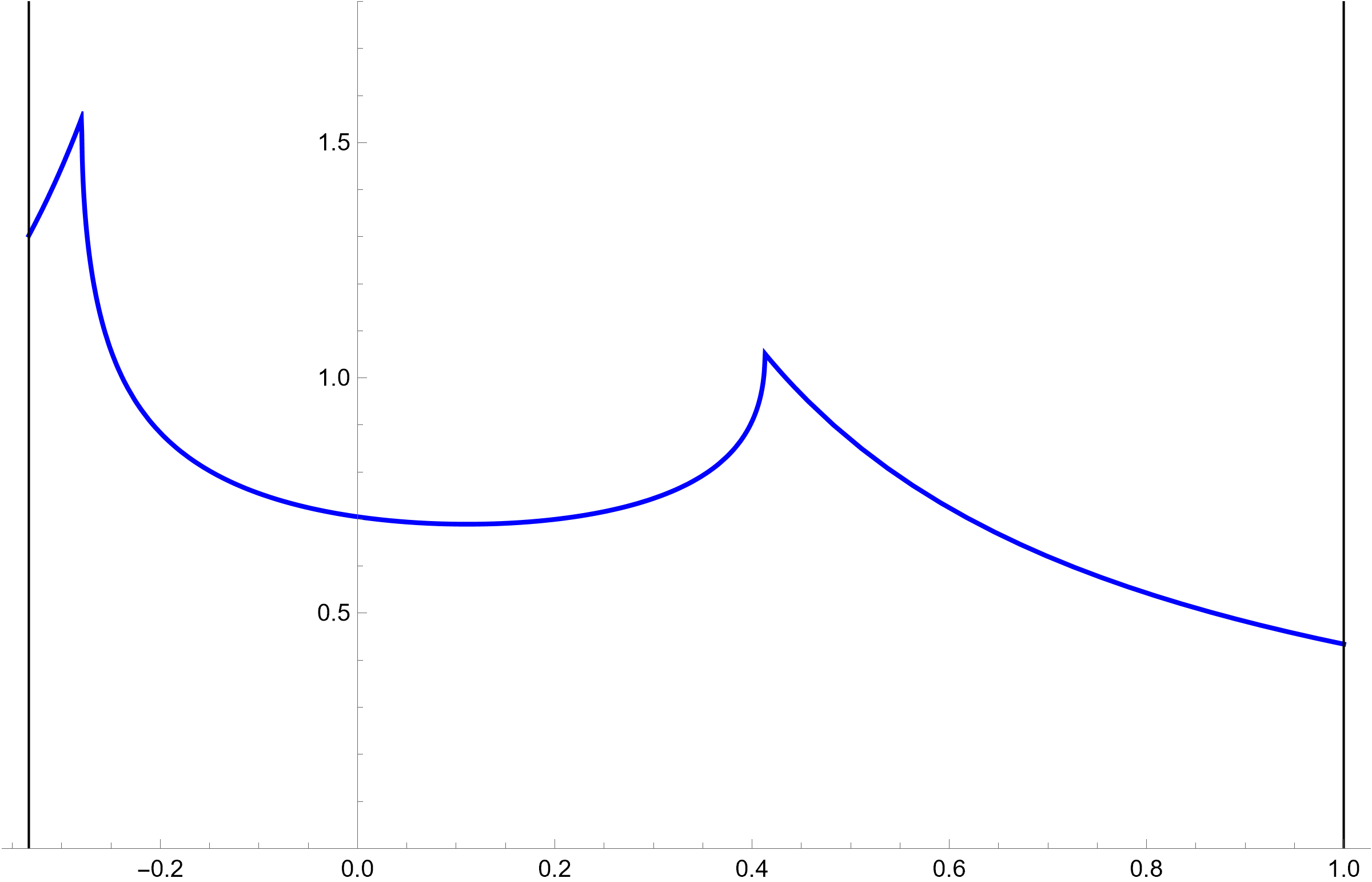}
		\end{center}
		\subcaption{ $\log(10)$}
	\end{subfigure}

	\begin{subfigure}{0.3\textwidth}
		\begin{center}	
		    \includegraphics[width=\textwidth]{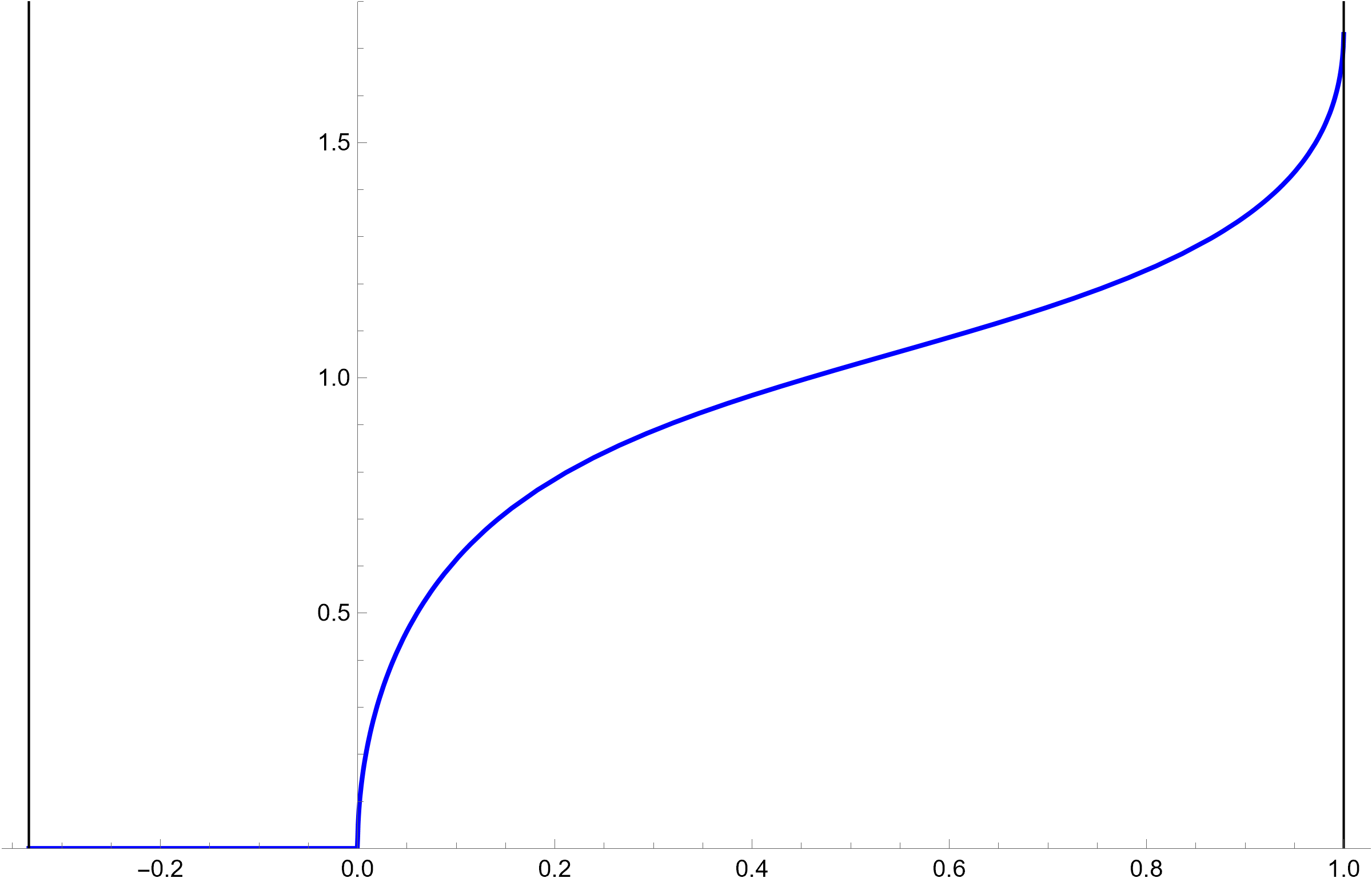}
		\end{center}
		\subcaption{$\log(4/3)$}
	\end{subfigure}
\quad
		\begin{subfigure}{0.3\textwidth}
		\begin{center}	
	  \includegraphics[width=\textwidth]{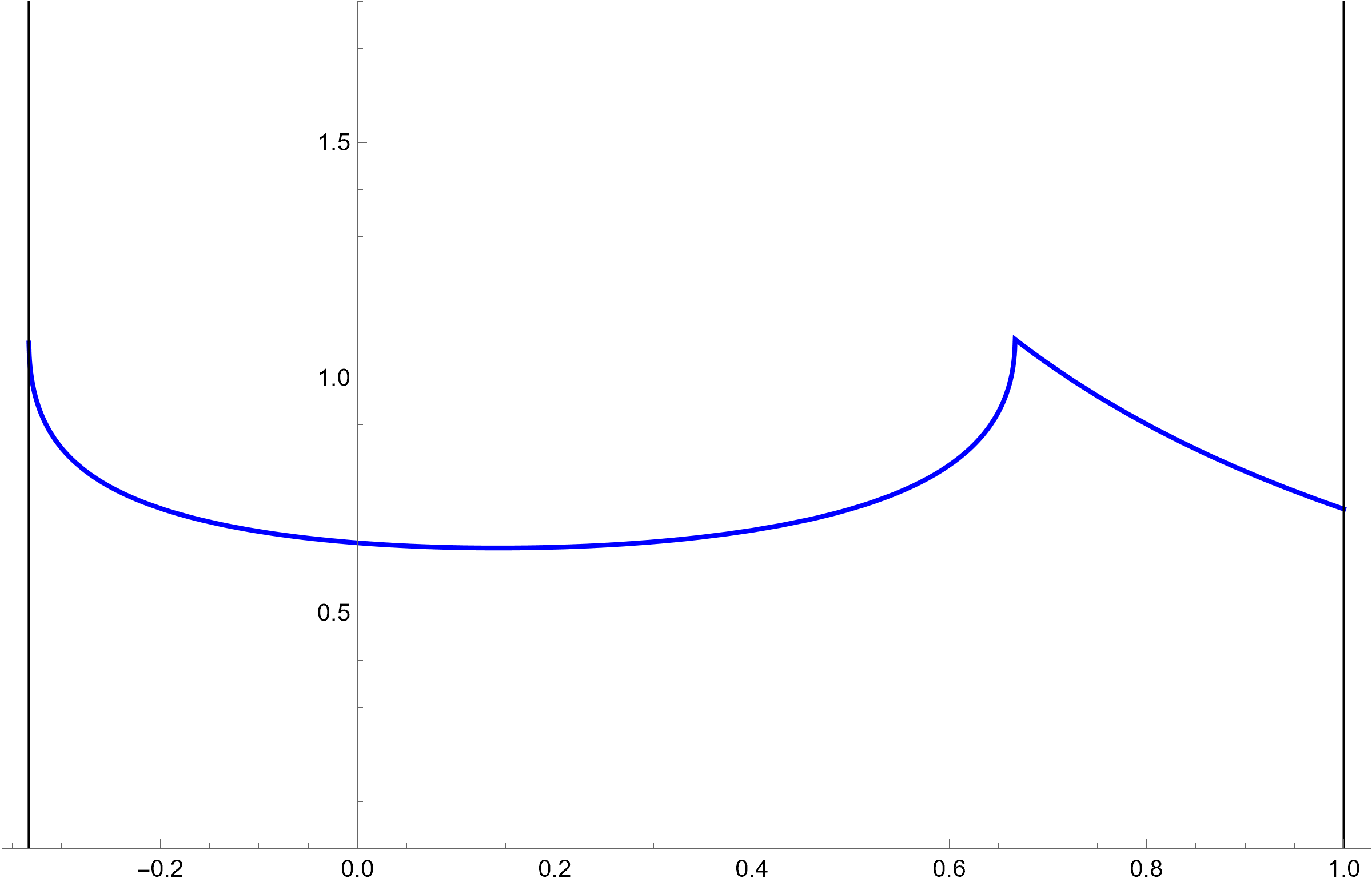}
		\end{center}
		\subcaption{$\log(4)$}
	\end{subfigure}	  
	\caption{ The plots display the density $x \mapsto \rho^{(a)}(x)$, given by \eqref{def of limiting density}, for $x \in [a,1]$ with $a = -1/3$. The critical regimes correspond to the values $\lambda = \log(4/3)$ and $\lambda = \log 4$, shown in (D) and (E) of the second row, respectively.
 } \label{Fig_limiting density}
\end{figure}

\begin{rem}[Spectral moments of the $q$-deformed GUE \cite{BFO24}: the special case $a=-1$] 
We discuss the special case $a=-1$ of our theorems. Due to the term $(a+1)^{p-2k}$ in the second summation of \eqref{moment closed main}, the summation becomes non-trivial only when $k=p/2$ for $p$ even. Then by \eqref{def of H special cases}, it follows that 
\begin{equation}
m_{N,2p}^{ (-1) } =   (1-q)^p \sum_{j=0}^{N-1}    \sum_{l=0}^k     q^{-l(2p-l)+\frac{l(l-1)}{2}}\frac{[2p]_{q}!}{[2p-2l]_{q}!! \, [l]_{q}!} q^{j(2p-l)}\qbinom{j}{l}.
\end{equation}
This recovers the recent finding \cite[Theorem 2.1]{BFO24}.  
Similarly, it follows from \eqref{def of mathcal Mp0} and \eqref{def of mathcal Mp1} that 
\begin{align}
\mathcal{M}_{2p,0}\Big|_{ a=-1 } &= \frac{1}{\lambda\, p}  I_{1-\mathsf{s} }(p+1,p),
\\
 \mathcal{M}_{2p,1}\Big|_{ a=-1 }  &=-\frac{\lambda p}{6}
          \bigg(  I_{1-\mathsf{s}  }(p+1,p)
        +\frac{(2p-1)!}{p!\,(p-1)!}  \mathsf{s}^{p} (1-\mathsf{s}  )^{p-1}(2+p-(2p+1) \mathsf{s} )\bigg).  
\end{align}
This recovers \cite[Theorem 2.3]{BFO24}.
On the other hand, Theorem~\ref{Thm_limiting density} generalises the finding of \cite[Theorem 2.6]{BFO24}, where, in the special case $a = -1$, the two transitions coincide, yielding only two distinct phases. To be more precise, it immediately follows from Theorem~\ref{Thm_limiting density} that  
\begin{align}
\begin{split}
\rho^{(-1)}(x) &= \frac{2}{ \pi \lambda |x| } \arctan \sqrt{ \frac{ 1-\sqrt{1-x^2} }{ 1+\sqrt{1-x^2}  }  \frac{ \sqrt{1-x^2}+1-2e^{-\lambda}  }{  \sqrt{1-x^2}-1+2e^{-\lambda}  }  }   \mathbbm{1}_{(-b(\lambda),b(\lambda))}(x) 
\\
&\quad +
\begin{cases}
0 &\textup{if } \lambda \le \log 2,
\smallskip 
\\
\displaystyle 
\frac{1}{\lambda |x|} \Big( \mathbbm{1}_{ (-1,1) }(x) -\mathbbm{1}_{ (-b(\lambda),b(\lambda)) }(x) \Big) &\textup{if } \lambda > \log 2 ,
\end{cases}
\end{split}
\end{align}
where $b(\lambda) :=  2 \sqrt{ (1-e^{ -\lambda }) e^{-\lambda} } $.  
This recovers the limiting density in \cite[Theorem 2.6]{BFO24}. 
\end{rem}

\begin{rem}[Continuum limit]
    Both Theorems~\ref{Thm_genus expansion} and~\ref{Thm_limiting density} are valid in the regime where $q \to 1$ while $a$ remains fixed. Nevertheless, one can formally consider the continuum limit of $\mathcal{M}_{p,0}$ by allowing $a \to -1$ under a suitable scaling, consistent with the continuum limit of the Al-Salam--Carlitz weight \eqref{limit of weight q to 1}. In particular, under the scaling $a = -1 + r\sqrt{\lambda}$, it follows from \eqref{def of Mp0 v2} that
    \begin{equation}
     \lim_{ \lambda \to 0 }   \frac{1}{\lambda^{p/2}}\mathcal{M}_{p,0} \Big|_{ a = -1 + r\sqrt{\lambda} }
        =\sum_{l=0}^{\floor{p/2}}\binom{p}{2l}r^{p-2l} C_{l}, 
    \end{equation}
    where $C_l$ is the $l$-th Catalan number. Here, the right-hand side of the equation can be interpreted as the $p$-th moment of the shifted semicircle law $\rho_{\mathrm{sc}}(x - r)$. This is consistent with the limiting behaviour \eqref{limit of weight q to 1} of the weight function.

    This phenomenon can also be observed at the level of the limiting spectral density. 
To analyse the behaviour as $\lambda \to 0$, we consider the following rescaling
    \begin{equation}
    \widehat{\rho}^{(a)}(y):= \sqrt{\lambda} \rho^{(a)}( \sqrt{\lambda} y ). 
    \end{equation}
    Then for $a \in [-1,0)$ and $\lambda < \log(1-a)$, by Theorem~\ref{Thm_limiting density}, we have 
    \begin{align}
   \widehat{\rho}^{(a)}(y) = \frac{2}{\pi\lambda\abs{y}}\arctan{\sqrt{\frac{1-y_0-y_1}{1-y_0+y_1}\frac{1-e^{-\lambda}-y_0+y_1}{y_0+y_1-1+e^{-\lambda}}}} \, \mathbbm{1}_{(\mathsf{u}-\mathsf{v},\mathsf{u}+\mathsf{v} ) }( \sqrt{\lambda} y ), 
\end{align}
    where
    \begin{equation} \label{def of y0 y1}
        y_{0}:=\frac{-\sqrt{\lambda}y(a+1)+a^{2}+1}{(a-1)^{2}},\qquad y_{1}:=\frac{\sqrt{4a(\sqrt{\lambda}y-a)(\sqrt{\lambda}y-1)}}{(a-1)^{2}} .
    \end{equation}
    Note that with scaling $a=-1+r \sqrt{\lambda}$, we have 
    \begin{align*}
    y_0= \frac12+\frac1{8} ( r^2-2r y ) \lambda +O( \lambda^{3/2}), \qquad  y_0= \frac12-\frac1{8} ( r^2-2r y +2y^2 ) \lambda +O( \lambda^{3/2}),
    \end{align*}
    as $\lambda \to 0$. This in turn implies that for $|y-r|<2$, 
    \begin{equation}
  \lim_{ \lambda \to 0 }  \widehat{\rho}^{(a)}(y) \Big|_{ a = -1 + r\sqrt{\lambda} }
  = \rho_{\mathrm{sc}}(y - r).  
    \end{equation}
    Again, this is consistent with the limiting behaviour \eqref{limit of weight q to 1}. 
\end{rem}

We conclude this section by discussing the appearance of the limiting spectral density $\rho^{(a)}$ from Theorem~\ref{Thm_limiting density} in connection with the zeros of Al-Salam--Carlitz polynomial \eqref{def of three term AlSalam}. In the continuum case when $q = 1$, it is a well-known fact that the limiting spectral distribution of a random Hermitian unitary ensemble coincides with the limiting zero distribution of the corresponding orthogonal polynomial. From the perspective of statistical physics, this correspondence reflects the interpretation of the zeros of orthogonal polynomials as Fekete points, or equivalently, as the zero-temperature ($\beta = \infty$) limit of the $\beta$-ensemble. The key feature is that the limiting equilibrium measure of the $\beta$-ensemble is independent of the specific value of $\beta$, except in the high-temperature regime where $\beta \to 0$. (We refer the reader to \cite{ABG12} for the high-temperature crossover in the classical log-gas setting, and to the recent work \cite{CD25} for its discrete analogue.)

In contrast to the classical ($q=1$) setting, no general result is known regarding the coincidence of the limiting spectral and zero distributions in the $q$-deformed case. However, for our specific setting, a relevant result is established in \cite[Proposition 4.2]{BF25a}, where it is shown that
\begin{equation}
\mathbb{E}\bigg[\prod_{j=1}^N (x - x_j)\bigg] = U_N^{(a)}(x).
\end{equation}
Here, the expectation is taken with respect to the joint density \eqref{def of jpdf}. This identity suggests that the limiting zero distribution of the Al-Salam--Carlitz polynomial $U_N^{(a)}$ coincides with the limiting spectral distribution described in Theorem~\ref{Thm_limiting density}. In the following proposition, we confirm that this is indeed the case, cf. Figure~\ref{Fig_diagram for OP}.

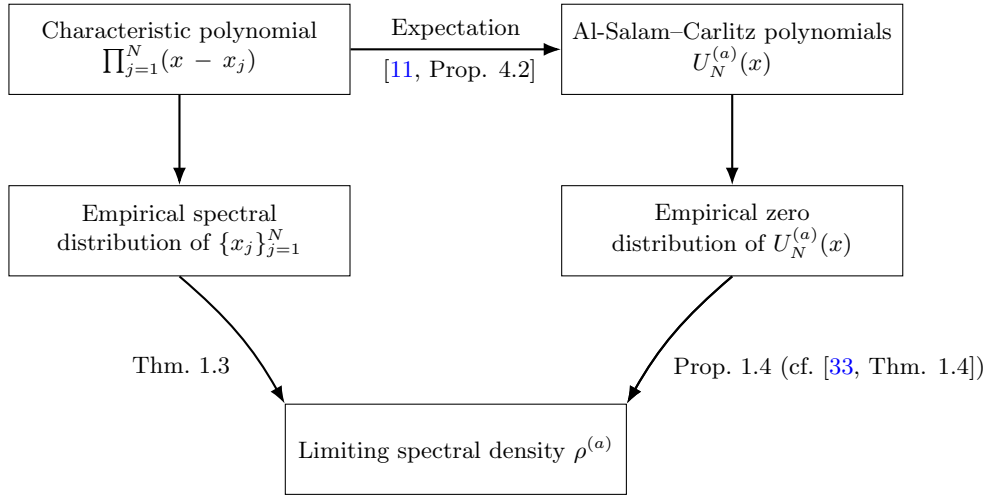
\begin{figure}[ht]
\centering
\begin{tikzpicture}[
  box/.style={
    draw, rectangle, minimum height=1.2cm, text width=4.3cm, align=center
  },
  arr/.style={-Latex, thick},
  node distance=1.2cm and 2.8cm,
  every node/.style={font=\small}
  ]
 
\node[box] (guepoly) {Characteristic polynomial \\ $ \prod_{j=1}^N  (x-x_j) $};
\node[box, right=of guepoly] (hermite) {Al-Salam--Carlitz polynomials \\ $U_N^{(a)}(x)$  };
 
\draw[arr] (guepoly.east) -- node[above]{Expectation}  node[below]{ \cite[Prop. 4.2]{BF25a} }(hermite.west);
 
\node[box, below=of guepoly] (gueemp) {Empirical spectral \\ distribution of $\{x_j\}_{j=1}^N$};
\node[box, below=of hermite] (hermzero) {Empirical zero \\ distribution of $U_N^{(a)}(x)$ };
 
\draw[arr] (guepoly.south) -- (gueemp.north);
\draw[arr] (hermite.south) -- (hermzero.north);
 
\path (gueemp) -- (hermzero) coordinate[midway] (midpoint);
 
\node[box, below=2.3cm of midpoint] (limit) 
{Limiting spectral density $\rho^{(a)}$};
 
\draw[arr] 
  (gueemp.south) to[bend left=10] 
  node[pos=0.6, below left, align=center] 
  { Thm.~\ref{Thm_limiting density} } 
  (limit.north west);

\draw[arr] 
  (hermzero.south) to[bend right=10] 
  (limit.north east);

\draw[arr] 
  (hermzero.south) to[bend right=10] 
  node[pos=0.6, below right, align=center]
  { Prop.~\ref{prop:Asymptotic zero distribution} (cf. \cite[Thm. 1.4]{KA99}) } 
  (limit.north east);

\end{tikzpicture}

\caption{ Schematic diagram illustrating the dual route to the limiting spectral density $\rho^{(a)}$ via characteristic polynomials and Al-Salam–Carlitz orthogonal polynomials.} \label{Fig_diagram for OP}
\end{figure}

To state the result, we define the empirical zero distribution of the Al-Salam--Carlitz polynomial by
\begin{equation} \label{def of ESD for AlSalam}
\nu_N := \frac{1}{N} \sum_{j=1}^N \delta_{x_{N,j}},
\end{equation}
where ${x_{N,1}, \dots, x_{N,N}}$ are the zeros of $U_N^{(a)}(x;q)$ and $\delta_x$ denotes the Dirac mass at $x$. Then we have the following.

\begin{prop}[\textbf{Limiting zero distribution of Al-Salam--Carlitz polynomial}]\label{prop:Asymptotic zero distribution}
Let $a < 0$ be fixed, and let $q$ be scaled according to \eqref{def of q scaling}. Let $\nu_N$ be the empirical zero distribution \eqref{def of ESD for AlSalam}. As $N \to \infty$, in the sense of integration against continuous test functions $f \in C([a,1])$, we have  
\begin{equation}
d\nu_N(x) \to  \rho^{(a)}(x)\,dx,
\end{equation} 
where $\rho^{(a)}$ is given in Theorem~\ref{Thm_limiting density}. 
\end{prop}

This proposition is largely independent of the preceding analysis leading to Theorem \ref{Thm_limiting density}, which derives the moments of the limiting spectral density via an algebraic--combinatorial approach. Instead, it is based on a general result of Kuijlaars and Van Assche~\cite[Theorem 1.4]{KA99}, obtained through potential-theoretic methods, concerning the limiting zero distribution of orthogonal polynomials---applicable also in the \( q \)-deformed setting. See also~\cite{DM98} for related developments in the continuum limit of the Toda lattice.
By applying~\cite[Theorem 1.4]{KA99} to our setting under appropriate scaling, we verify that the same limiting density \( \rho^{(a)} \) arises as the zero distribution. It is particularly noteworthy that the same limiting object emerges from two seemingly distinct perspectives. Moreover, as exemplified by the Tricomi--Carlitz polynomials in~\cite{KA99}, the appearance of a non-smooth peak is a typical feature of the limiting zero distribution of \( q \)-orthogonal polynomials in the double scaling limit \( q \to 1 \).
For further developments on the connection between recurrence coefficients and limiting zero distributions---particularly their interpretation via natural operations in free probability---we refer the reader to the recent works~\cite{JKM25,AFPU24} and the references therein.

\medskip

\subsection*{Organisation of the paper} The remainder of the paper is organised as follows. In the next section, we present the necessary background for our analysis, including the Al-Salam--Carlitz orthogonal polynomials on discrete spaces and the Flajolet--Viennot theory. In Section~\ref{Section_spectral moments}, we prove Theorem~\ref{Thm_spectral moments} via a double counting argument, identifying the enumeration of Motzkin paths with a certain matching problem that allows for explicit evaluation. Section~\ref{Section_largeN asymptotic} is devoted to the large-\( N \) asymptotic analysis, where we establish Theorems~\ref{Thm_genus expansion} and~\ref{Thm_limiting density}. In addition, we provide the proof of Proposition~\ref{prop:Asymptotic zero distribution}.

\medskip

 \subsection*{Acknowledgements} We are grateful to Peter J.~Forrester for his valuable feedback on an earlier draft of this work. Sung-Soo Byun was supported by the National Research Foundation of Korea grant (RS-2023-00301976, RS-2025-00516909). Jaeseong Oh was partially supported by a KIAS Individual Grant (HP083401) via the Center for Mathematical Challenges at Korea Institute for Advanced Study.

\medskip

\section{Preliminaries}\label{Section_Prelim}

In this section, we recall the background material required for the proof of our main results.

\subsection{Al-Salam--Carlitz polynomials} \label{Subsection_Al-Salam poly}

The Al-Salam--Carlitz polynomials $U_n^{(a)}(x)$ form families of orthogonal polynomials in the discrete space. Let us first recall the notion of the Jackson $q$-integral, which is defined by
\begin{align}
\int_0^\infty f(x)\,d_qx  :=(1-q) \sum_{k=-\infty}^\infty q^k \, f(q^k) ,
\qquad 
\int_0^\alpha f(x)\,d_qx  :=(1-q)\sum_{k=0}^\infty \alpha  q^k\, f(\alpha q^k),
\end{align}
and 
\begin{equation}
\int_a^b f(x)\,d_qx  := \int_0^b f(x)\,d_q x - \int_0^a f(x)\,d_q x. 
\end{equation} 
Then the orthogonality of Al-Salam Carlitz polynomials reads as
\begin{equation} \label{def of orthogonality AlSalam}
    \int_{a}^{1}(qx,qx/a;q)_{\infty}U_{m}^{(a)}(x;q)U_{n}^{(a)}(x;q) \, d_{q}x
    =(-a)^{n}(1-q)(q;q)_{n}(q,a,q/a;q)_{\infty}q^{\frac{n(n-1)}{2}}\delta_{m,n}, 
\end{equation}
where $\delta$ is the Kronecker delta, see e.g. \cite[Section 14.24]{KLS10}.

According to \eqref{def of AlSalam weight hat}, it is also convenient to consider the scaling 
\begin{equation} \label{def of Un hat}
\widehat{U}_n^{(a)}(x;q):= \frac{1}{(1-q)^{n/2}} U_n^{(a)} ( \sqrt{1-q}x;q ). 
\end{equation}
Then it follows from \eqref{def of three term AlSalam} that 
\begin{equation} \label{def of three term AlSalam rescaling}
x\,\widehat{U}_n^{(a)}(x;q) = \widehat{U}_{n+1}^{(a)}(x;q) +\frac{a+1}{ \sqrt{1-q} }\, q^n \widehat{U}_n^{(a)}(x;q) -a q^{n-1} [n]_q \widehat{U}_{n-1}^{(a)}(x;q). 
\end{equation}
On the other hand, for general \( a <0 \),  we have 
\begin{equation}
\lim_{q \to 1} \widehat{U}_n^{(a)}(x;q)\Big|_{ a=-q^r } = He_n(x-r),   
\end{equation}
where $He_n$ is the classical Hermite polynomial, see \cite[Eq.~(2.11)]{BF00}.  

\medskip 

We write 
\begin{equation}  
\mathfrak{m}_{p,j}^{(a)}:= \int_a^1 \frac{ x^p \, U_{j}^{(a)}(x;q)^2  }{ (-a)^{j}(1-q)(q;q)_{j} q^{\frac{j(j-1)}{2}} }  \,\omega_U^{(a)}(x) \,d_qx,
\end{equation}
where the weight $\omega_U^{(a)}$ is given by \eqref{def of AlSalam weight}.
Then by \eqref{def of 1pt density} and \eqref{def of moment in terms of q-integral}, it follows that 
\begin{equation} \label{m pj as a sum of mathfrak m}
m_{N,p}^{ (a) } =  \sum_{j=0}^{N-1} \mathfrak{m}_{p,j}^{(a)}. 
\end{equation}
In addition, we define the rescaled moment
\begin{equation} \label{def of mathfrak m pj hat}
\widehat{\mathfrak{m}}_{p,j}^{(a)} := \frac{1}{(1-q)^{p/2}} \mathfrak{m}_{p,j}^{(a)},
\end{equation}
which corresponds to the scaling in \eqref{def of Un hat}. This normalisation arises from the change of variables in the $q$-Jackson integral:
\begin{equation}\label{1.43}
\int_{-1}^1 f(x) \,d_q x = \frac{1}{c} \int_{-c}^c f\Big( \frac{x}{c} \Big) \, d_q x.
\end{equation}

\subsection{Flajolet-Vienot theory}

To evaluate the spectral moments, we employ a combinatorial approach developed by Flajolet and Viennot \cite{Fl80, Vi00}, which connects the moments of orthogonal polynomials to weighted lattice paths, see also the review \cite{CKS16}.  
This perspective is rooted in algebraic combinatorics, where the enumeration of lattice paths--such as Motzkin paths--provides a classical framework for interpreting various algebraic and analytic quantities.  
In this subsection, we briefly recall the relevant aspects of their theory.

Let $\mathcal{L}:\R[x]\rightarrow\R$ be a linear functional associated with the weighted Lebesgue measure $\omega(x)\,dx$, i.e. 
\begin{equation*}
    \mathcal{L}f=\int_{\R}f(x)\omega(x)\, dx, \qquad f\in\R[x].
\end{equation*}
Let $\set{P_j}_{j\geq0}$ be the sequence of monic orthogonal polynomials with respect to $\mathcal{L}$, namely,
\begin{equation}
\mathcal{L}( P_j(x) P_k(x) ) = \mathcal{L}( P_j(x) P_j(x) ) \, \delta_{j,k}. 
\end{equation}
Then it is known that $\set{P_{j}}_{j\geq0}$ satisfies the three-term recurrence relation of the form: 
\begin{equation} \label{def of three term general}
    P_{n+1}(x)=(x-b_{n})P_{n}(x)-\lambda_{n}P_{n-1}(x).
\end{equation} 

Next, we recall the class of lattice paths used in the Flajolet--Viennot theory. A lattice path $\omega$ is a finite sequence $\omega = (s_0, \dots, s_n)$, where each $s_j = (x_j, y_j)$ belongs to $\mathbb{N} \times \mathbb{N}$. A step $(s_j, s_{j+1})$ is called:
\begin{itemize}
    \item \emph{North-East} if $(x_{j+1}, y_{j+1}) = (x_j + 1, y_j + 1)$;
    \smallskip 
    \item \emph{East} if $(x_{j+1}, y_{j+1}) = (x_j + 1, y_j)$;
       \smallskip 
    \item \emph{South-East} if $(x_{j+1}, y_{j+1}) = (x_j + 1, y_j - 1)$.
\end{itemize}

We say that a step $(s_j, s_{j+1})$ is at height $k$ if $y_j = k$. 
A \emph{Motzkin path} is a path composed of North-East, East, and South-East steps that remain within the first quadrant.
Let $\Mot_{n,k,l}$ denote the set of Motzkin paths from the point $(0, k)$ to $(n, l)$. Given two sequences $b = \{ b_n \}_{n \geq 0}$ and $\lambda = \{ \lambda_n \}_{n \geq 0}$, and a Motzkin path $\omega \in \Mot_{n,k,l}$, we define its associated weight $\wt_{b,\lambda}(\omega)$ as the product of the weights of its individual steps, assigned as follows:
\begin{itemize}
    \item Each North-East step has weight $1$;
    \smallskip 
    \item An East step at height $k$ has weight $b_k$;
     \smallskip 
    \item A South-East step at height $k$ has weight $\lambda_k$.
\end{itemize}
Then the following identity, arising from the theory of Flajolet \cite{Fl80} and Viennot \cite{Vi00} admits the combinatorial interpretation of the spectral moments: 
\begin{equation} \label{def of Mot path counting gen}
 \frac{  \mathcal{L}(  x^p \,P_j(x)^2) }{  \mathcal{L}( P_j(x)^2 )  } =\sum_{\omega\in\mathrm{Mot}_{p,j,j}}\mathrm{wt}_{b,\lambda}(w).  
\end{equation}
See also \cite[Proposition 3.2]{BFO24} and the references therein. This identity shows that the spectral moments can be expressed as weighted sums over all Motzkin paths, where the weights are determined by the coefficients \( b_n \) and \( \lambda_n \) in the three-term recurrence relation \eqref{def of three term general}.

We remark that the right-hand side of \eqref{def of Mot path counting gen} serves as the partition function (normalisation constant) for defining a natural probability measure on the space of Motzkin paths. The statistical properties of random Motzkin paths and their connection to the KPZ universality class have been investigated in recent work~\cite{BKW25}.

\section{Combinatorics for spectral moments} \label{Section_spectral moments}

In this section, we show Theorem~\ref{Thm_spectral moments}. Note that by \eqref{m pj as a sum of mathfrak m} and \eqref{def of mathfrak m pj hat}, it suffices to show that  \begin{equation}
 \widehat{\mathfrak{m}}^{(a)}_{p,j} = \sum_{k=0}^{\floor{p/2}}\sum_{l=0}^k \Big(\frac{a+1}{\sqrt{1-q}}\Big)^{p-2k}(-a)^{k}q^{-l(p-l)+\frac{l(l-1)}{2}}\frac{[p]_{q}!}{[p-2l]_{q}!![l]_{q}!}\mathsf{H}(k-l,p-2k)q^{j(p-l)}\qbinom{j}{l} . \label{eqn:mpj}  
\end{equation}
On the other hand, it follows from \eqref{def of Mot path counting gen} that  
\begin{equation}  \label{eqn:m v0}
 \widehat{\mathfrak{m}}_{p,j}^{(a)} = \sum_{\omega\in\mathrm{Mot}_{p,j,j}}\mathrm{wt}_{b,\lambda}(\omega), 
\end{equation}
where the sequences $\set{b_n}_{n\geq0}$ and $\set{\lambda_n}_{n\geq0}$ are determined by 
\begin{equation} \label{def of three term coeff AlSalam}
    b_n=\frac{a+1}{\sqrt{1-q}}q^n, \qquad \lambda_{n}=-aq^{n-1}[n]_q,
\end{equation} 
see \eqref{def of three term AlSalam rescaling}. To analyse the right-hand side of \eqref{eqn:m v0}, we first decompose the set $\mathrm{Mot}_{p,j,j}$ according to the number of East steps. Let $\mathrm{Mot}_{p,j,j}^{k}$ denote the subset of Motzkin paths from $(0,j)$ to $(p,j)$ consisting of $k$ North-East and $k$ South-East steps, and $p - 2k$ East steps. Then by definition, we have 
\begin{equation}\label{eqn:m}
\widehat{\mathfrak{m}}_{p,j}^{(a)} = \sum_{k=0}^{\lfloor p/2 \rfloor} \sum_{\omega \in \mathrm{Mot}_{p,j,j}^{k}} \mathrm{wt}_{b,\lambda}(\omega).
\end{equation}
Combining the ingredients above, Theorem~\ref{Thm_spectral moments} is reduced to showing the following proposition.

\begin{prop} \label{Prop_spectral moments combinatorics}
For nonnegative integers $k$, $p$, and $j$, we have 
\begin{equation}\label{eqn:spectral moment via weight}
\sum_{\omega\in\mathrm{Mot}^{k}_{p,j,j}}\mathrm{wt}_{b,\lambda}(\omega) = \sum_{l=0}^k \Big(\frac{a+1}{\sqrt{1-q}}\Big)^{p-2k}(-a)^{k}q^{-l(p-l)+\frac{l(l-1)}{2}}\frac{[p]_{q}!}{[p-2l]_{q}!![l]_{q}!}\mathsf{H}(k-l,p-2k)q^{j(p-l)}\qbinom{j}{l} . 
\end{equation}
\end{prop}

The rest of this section is devoted to proving this proposition. 

\subsection{Matching problem} Due to the complexity of direct computation in the Motzkin path framework, we introduce an alternative combinatorial interpretation of (\ref{eqn:m}), namely in terms of \textit{matching}, to make the calculation tractable.

For a positive integer $n$, consider a set $[n]=\set{1,2,\cdots,n}$. Each element $x\in [n]$ is called a \textit{vertex}. A \textit{matching} on $[n]$ is a set partition of $[n]$ such that each part has size $1$ or $2$. The singleton parts are called \textit{isolated vertices} and the parts consisting of $2$ elements are called \textit{arcs}. Given an arc $(a,b)$ with $a<b$, the vertex $a$ is called an opener, and the vertex $b$ is called a closer.

We extend the notion of matching by distinguishing two types of isolated vertices: isolated vertices and \textit{verticals}. The resulting object, in which isolated vertices are assigned types, will be referred to as a \textit{generalised matching}. Let us denote by ${\mathrm{Mat}}_{a,b,c}$ the set of generalised matchings on $[a]$ consisting of $b$ arcs and $c$ vertical lines. A generalised matching on $[n]$ has an pictorial representation as a diagram consists of $n$ vertices placed in a horizontal row from left to right, labeled $1$ through $n$, some of them are connected by arcs, and some of them have a vertical line on them.

\begin{figure}[h]
\centering
\begin{tikzpicture}[scale=0.8]
  \foreach \i in {0,2,4,6,8,10,12} {\fill (\i,0) circle (2pt);}
  \draw (0,0) to[out=60,in=120] (4,0);
  \draw (2,0) to[out=40,in=140] (10,0);
  \draw[] (8,0) -- +(0,2);
\end{tikzpicture}
\caption{Pictorial representation of generalised matching on $[7]$ with 2 arcs and 1 vertical}
\end{figure}
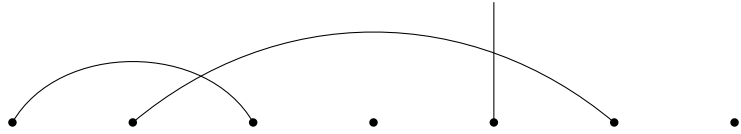

To proceed, let ${\mathrm{Mat}}_{a,b,c}^{>j}$ be the subset of ${\mathrm{Mat}}_{a,b,c}$ consisting of generalised matching of which the first $j$ vertices are either isolated or openers. On the other hand, we define a \textit{Al-Salam--Carlitz history} as a labeled Motzkin path where the South-East step of height $k$ is labeled with an integer in $[k]$. Then we proceed to construct a bijection between Al-Salam--Carlitz histories on $\mathrm{Mot}_{p,j,j}^{k}$ and generalised matchings in ${\mathrm{Mat}}_{p+j,k,p-2k}^{>j}$ in the following manner.
\begin{itemize}
    \item Start with $j$ vertices.
    \smallskip 
    \item Progressing through each step of the Motzkin path from $(0,j)$, a vertex is added.
      \smallskip 
    \item If the Al-Salam--Carlitz history advances with a South-East step labeled by $m$, we introduce a new arc connecting the newly added vertex and the $m-$th isolated vertex, counted from the left.
      \smallskip 
    \item If the Al-Salam--Carlitz history advances with a East step, we introduce a vertical on a newly added vertex.
      \smallskip 
    \item If the Al-Salam--Carlitz history advances with a North-East step, a vertex is added without any accompanying arc.
\end{itemize}
See Figure \ref{fig:comparison_models} for a pictorial illustration of each step.

\begin{figure}
\centering

\begin{minipage}{0.45\textwidth}
\centering
\begin{tikzpicture}[scale=0.7]
    \draw[step=1cm,gray!30,very thin] (0,0) grid (6,4);
    \draw[very thick, blue] (0,0) -- (0,3) ;
\end{tikzpicture}
\end{minipage}
\hspace{0.05\textwidth}
\begin{minipage}{0.45\textwidth}
\centering
\begin{tikzpicture}[scale=0.7]
    \foreach \i in {1,2,...,10}{\fill[black] (\i, 0) circle (0pt);}
    \foreach \i in {1,2,3}{\fill[blue] (\i, 0) circle (2pt);}
    \fill [black] (1,1.5) circle (0pt);
\end{tikzpicture}
\end{minipage}

\vspace{2mm}

\begin{minipage}{0.45\textwidth}
\centering
\begin{tikzpicture}[scale=0.7]
    \draw[step=1cm,gray!30,very thin] (0,0) grid (6,4);
    \draw[very thick] (0,0) -- (0,3);
    \draw[very thick, blue] (0,3) -- (1,2);
    \node[blue] at (1/4,9/4) {2};
\end{tikzpicture}
\end{minipage}
\hspace{0.05\textwidth}
\begin{minipage}{0.45\textwidth}
\centering
\begin{tikzpicture}[scale=0.7]
    \foreach \i in {1,2,...,10}{\fill[black] (\i, 0) circle (0pt);}
    \foreach \i in {1,2,3}{\fill[black] (\i, 0) circle (2pt);}
    \fill [black] (1,1.5) circle (0pt);
    \fill [blue] (4,0) circle (2pt);
    \draw [blue] (2,0) to[out=60,in=120] (4,0);
\end{tikzpicture}
\end{minipage}

\vspace{2mm}

\begin{minipage}{0.45\textwidth}
\centering
\begin{tikzpicture}[scale=0.7]
    \draw[step=1cm,gray!30,very thin] (0,0) grid (6,4);
    \draw[very thick] (0,0) -- (0,3) -- (1,2);
    \node at (1/4,9/4) {2};
    \draw[very thick, blue] (1,2) -- (2,2);
\end{tikzpicture}
\end{minipage}
\hspace{0.05\textwidth}
\begin{minipage}{0.45\textwidth}
\centering
\begin{tikzpicture}[scale=0.7]
    \foreach \i in {1,2,...,10}{\fill[black] (\i, 0) circle (0pt);}
    \foreach \i in {1,2,3,4}{\fill[black] (\i, 0) circle (2pt);}
    \fill [black] (1,1.5) circle (0pt);
    \fill[blue] (5,0) circle (2pt);
    \draw (2,0) to[out=60,in=120] (4,0);
    \draw[blue] (5,0) -- +(0,1.5);
\end{tikzpicture}
\end{minipage}

\vspace{2mm}

\begin{minipage}{0.45\textwidth}
\centering
\begin{tikzpicture}[scale=0.7]
    \draw[step=1cm,gray!30,very thin] (0,0) grid (6,4);
    \draw[very thick] (0,0) -- (0,3) -- (1,2) -- (2,2);
    \node at (1/4,9/4) {2};
    \draw[very thick, blue] (2,2) -- (3,3);
\end{tikzpicture}
\end{minipage}
\hspace{0.05\textwidth}
\begin{minipage}{0.45\textwidth}
\centering
\begin{tikzpicture}[scale=0.7]
    \foreach \i in {1,2,...,10}{\fill[black] (\i, 0) circle (0pt);}
    \foreach \i in {1,2,3,4,5}{\fill[black] (\i, 0) circle (2pt);}
    \fill [black] (1,1.5) circle (0pt);
    \fill[blue] (6,0) circle (2pt);
    \draw (2,0) to[out=60,in=120] (4,0);
    \draw (5,0) -- +(0,1.5);
\end{tikzpicture}
\end{minipage}

\vspace{2mm}

\begin{minipage}{0.45\textwidth}
\centering
\begin{tikzpicture}[scale=0.7]
    \draw[step=1cm,gray!30,very thin] (0,0) grid (6,4);
    \draw[very thick] (0,0) -- (0,3) -- (1,2) -- (2,2) -- (3,3);
    \node at (1/4,9/4) {2};
    \draw[very thick, blue] (3,3) -- (4,3);
\end{tikzpicture}
\end{minipage}
\hspace{0.05\textwidth}
\begin{minipage}{0.45\textwidth}
\centering
\begin{tikzpicture}[scale=0.7]
    \foreach \i in {1,2,...,10}{\fill[black] (\i, 0) circle (0pt);}
    \foreach \i in {1,2,3,4,5,6}{\fill[black] (\i, 0) circle (2pt);}
    \fill [black] (1,1.5) circle (0pt);
    \fill [blue] (7,0) circle (2pt);
    \draw (2,0) to[out=60,in=120] (4,0);
    \draw (5,0) -- +(0,1.5);
    \draw [blue] (7,0) -- +(0,1.5);
\end{tikzpicture}
\end{minipage}

\vspace{2mm}

\begin{minipage}{0.45\textwidth}
\centering
\begin{tikzpicture}[scale=0.7]
    \draw[step=1cm,gray!30,very thin] (0,0) grid (6,4);
    \draw[very thick] (0,0) -- (0,3) -- (1,2) -- (2,2) -- (3,3) --(4,3);
    \node at (1/4,9/4) {2};
    \draw[very thick, blue] (4,3) -- (5,4);
\end{tikzpicture}
\end{minipage}
\hspace{0.05\textwidth}
\begin{minipage}{0.45\textwidth}
\centering
\begin{tikzpicture}[scale=0.7]
    \foreach \i in {1,2,...,10}{\fill[black] (\i, 0) circle (0pt);}
    \foreach \i in {1,2,3,4,5,6,7}{\fill[black] (\i, 0) circle (2pt);}
    \fill [black] (1,1.5) circle (0pt);
    \fill [blue] (8,0) circle (2pt);
    \draw (2,0) to[out=60,in=120] (4,0);
    \draw (5,0) -- +(0,1.5);
    \draw (7,0) -- +(0,1.5);
\end{tikzpicture}
\end{minipage}

\begin{minipage}{0.45\textwidth}
\centering
\begin{tikzpicture}[scale=0.7]
    \draw[step=1cm,gray!30,very thin] (0,0) grid (6,4);
    \draw[very thick] (0,0) -- (0,3) -- (1,2) -- (2,2) -- (3,3) --(4,3) -- (5,4);
    \node at (1/4,9/4) {2};
    \draw[very thick, blue] (5,4) -- (6,3);
    \node[blue] at (21/4,13/4) {2};
\end{tikzpicture}
\end{minipage}
\hspace{0.05\textwidth}
\begin{minipage}{0.45\textwidth}
\centering
\begin{tikzpicture}[scale=0.7]
    \foreach \i in {1,2,...,10}{\fill[black] (\i, 0) circle (0pt);}
    \foreach \i in {1,2,3,4,5,6,7,8}{\fill[black] (\i, 0) circle (2pt);}
    \fill [black] (1,1.5) circle (0pt);
    \fill [blue] (9,0) circle (2pt);
    \draw (2,0) to[out=60,in=120] (4,0);
    \draw (5,0) -- +(0,1.5);
    \draw (7,0) -- +(0,1.5);
    \draw [blue] (3,0) to[out=40,in=140] (9,0);
\end{tikzpicture}
\end{minipage}

\caption{An Al-Salam--Carlitz history from $(0,3)$ to $(6,3)$ and its corresponding generalised matching}
\label{fig:comparison_models}
\end{figure}
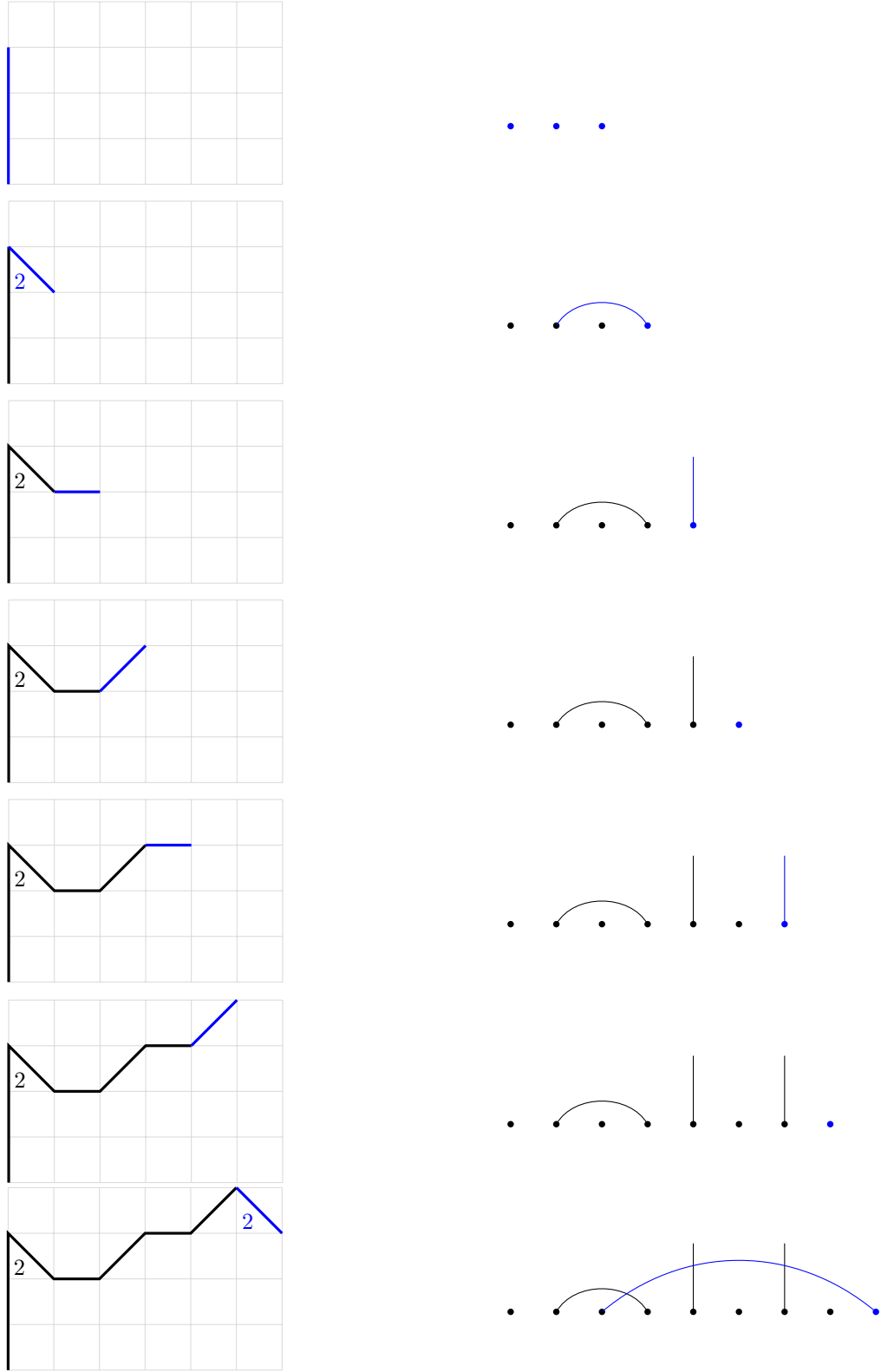

\medskip 
 
Now, we define some statistics for generalised matchings. 

\begin{defn}{Crossing}
A \textit{crossing} of a generalised matching $M$ is one of the following:
\begin{enumerate}
    \item a pair of arcs $(a,b)$ and $(c,d)$ with $a<c<b<d$,   \smallskip 
    \item a pair of an arc $(a,b)$ and an isolated vertex $c$ with $a<c<b$,   \smallskip 
    \item a pair of an arc $(a,b)$ and a vertical $c$ with $a<c<b$, or   \smallskip 
    \item a pair of an isolated vertex $a$ and a vertical $b$ with $a<b$.
\end{enumerate}
The total number of crossings in $M$ is denoted by $\mathrm{cr}(M)$. 
\end{defn}
Pictorially, the crossings are ``crossings'' made by arcs or verticals, after adding an additional vertex at ``infinity'' and connecting it with isolated vertices in the second and fourth cases.

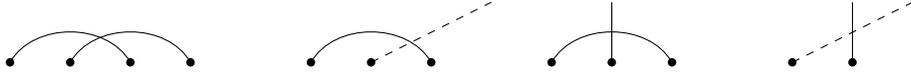
\begin{figure}[H]
\centering
\begin{tikzpicture}[scale=0.8]
  \foreach \i in {0,1,2,3,5,6,7,9,10,11,13,14} {\fill (\i,0) circle (2pt);}
  \draw (0,0) to [out=60,in=120] (2,0);
  \draw (1,0) to [out=60,in=120] (3,0);
  
  \draw (5,0) to [out=60,in=120] (7,0);
  \draw[dashed] (6,0) to (8,1);

  \draw (9,0) to [out=60,in=120] (11,0);
  \draw[] (10,0) -- +(0,1);

  \draw[] (14,0) -- +(0,1);
  \draw[dashed] (13,0) -- (15,1);
\end{tikzpicture}
\caption{Pictorial illustrations of four types of crossings}
\end{figure}

\begin{defn}{Nesting} 
A \textit{nesting} of a generalised matching $M$ is either
\begin{enumerate}
    \item a pair of arcs $(a,b)$ and $(c,d)$ with $a<c<d<b$, or
      \smallskip 
    \item a pair of an arc $(a,b)$ and an isolated vertex $c$ with $c<a<b$.
\end{enumerate}
The total number of nestings in $M$ is denoted by $\mathrm{ne}(M)$.
\end{defn}
Pictorially, the nestings are ``nestings'' made by two arcs, after adding an additional vertex at `infinity' and connecting it with isolated vertices in the second case.

\begin{figure}[H]
\centering
\begin{tikzpicture}[scale=0.8]
  \foreach \i in {0,1,2,3,5,6,7} {\fill (\i,0) circle (2pt);}
  \draw (0,0) to [out=60,in=120] (3,0);
  \draw (1,0) to [out=60,in=120] (2,0);
  
  \draw (6,0) to [out=60,in=120] (7,0);
  \draw[dashed] (5,0) to (8,1);

\end{tikzpicture}
\caption{Pictorial illustrations of two types of nestings}
\end{figure}
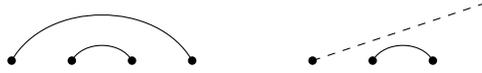

We now present an alternative enumeration of the left-hand side of \eqref{eqn:spectral moment via weight} by counting the crossings and nestings of generalised matchings. More precisely, we define a statistic on a generalised matching \( M \) by  
\begin{equation}\label{def of stat(M)}
    \mathrm{stat}(M):=\mathrm{cr}(M)+2\mathrm{ne}(M), 
\end{equation}
which will be used to enumerate the weight of Motzkin paths in the following lemma.

\begin{lem}
    For nonnegative integers $k$, $p$, and $j$, we have 
    \begin{equation}
    \sum_{\omega\in\mathrm{Mot}^{k}_{p,j,j}}\mathrm{wt}_{b,\lambda}(w)
    =
    \Big(\frac{a+1}{\sqrt{1-q}}\Big)^{p-2k}
    (-a)^{k}
    \sum_{M{\in}{\mathrm{Mat}}^{>j}_{p+j,k,p-2k}}q^{\mathrm{stat}(M)}.
    \end{equation}
\end{lem}

\begin{proof}
    For $w\in\mathrm{Mot}_{p,j,j}^{k}$, let $\mathrm{Mat}(w)$ be the set of generalised matchings corresponding to all possible Al-Salam--Carlitz histories on $w$. Then, due to the bijection between the Al-Salam--Carlitz histories on $\mathrm{Mot}_{p,j,j}^{k}$ and ${\mathrm{Mat}}^{>j}_{p+j,k,p-2k}$ described in Figure \ref{fig:comparison_models}, $\set{\mathrm{Mat}(w):w\in \mathrm {Mot}_{p,j,j}^{k}}$ is a set partition of $\mathrm{Mat}_{p+j,k,p-2k}^{>j}$. Thus, it suffices to establish that
    \begin{equation}\label{eqn:wtblambda(w)}
        \mathrm{wt}_{b,\lambda}(w)=\Big(\frac{a+1}{\sqrt{1-q}}\Big)^{p-2k}(-a)^{k}\sum_{M\in \mathrm{Mat}(w)}q^{\mathrm{stat}(M)}
    \end{equation}
    for each $w\in \mathrm{Mot}_{p,j,j}^{k}$. Recalling the construction of the bijection, observe that the number of isolated vertices changes by $+1$ for a North-East step, $0$ for an East step and $-1$ for a South-East step. Consequently, the height of each step indicates the number of isolated vertices in the generalised matching induced by previous steps. Based on this observation, we evaluate the contribution of each step to $q^{\mathrm{stat}(M)}$ according to the types of the step. To simplify, suppose that we have constructed the generalised matching $M_0$ so far and encountered the step of height $h$.
    \begin{enumerate}
        \item The North-East step has the weight $1$. On the generalised matching side, we add an isolated vertex at the right end, which contributes $0$ to the statistic. \smallskip 
        \item Suppose that we encounter the East step of height $h$. Then, on the generalised matching side, we add vertical which makes crossings with isolated vertices in $M_0$. Since the number of the isolated vertices is exactly $h$, this step contributes
        \begin{equation*}
            \frac{a+1}{\sqrt{1-q}}q^{h}
        \end{equation*}
        to the right-hand side of \eqref{eqn:wtblambda(w)}, which corresponds to the weight $b_h$. Here, note that the factor $\frac{a+1}{\sqrt{1-q}}$ does not depend on the height of the step. \smallskip 
        \item Finally, suppose that we encounter the South-East step labeled by $m$ with $1\leq m\leq h$. Then, on the generalised matching side, we make an arc $\alpha=(i_m, x)$ connecting the newly added vertex $x$ and the $m$-th isolated vertex $i_m$ in $M_0$. Then $\alpha$ forms crossing with one of the following: (1) the isolated vertices between $i_m$ and $x$, (2) vertical between $i_m$ and $x$, or (3) arcs whose opener is smaller than $i_m$ and closer lies between $i_m$ and $x$. Moreover, $\alpha$ also forms nestings with isolated vertices smaller than $i_m$, or arcs lying between $\alpha$. However, note that the statistics made with verticals or nestings are already counted in $\mathrm{stat}(M)$ as contributions of the isolated vertex $i_m$. Consequently, the newly added arc $\alpha$ contributes $(h-m)+2(m-1)$ to the statistic. Furthermore, there are $h$ possible label from $1$ to $h$. Therefore, this step contributes
        \begin{equation*}
            (-a) \sum_{m=1}^{h}q^{(h-m)+2(m-1)}=(-a)q^{h-1}[h]_{q}
        \end{equation*}
        to the right-hand side of \eqref{eqn:wtblambda(w)}, which corresponds to the weight $\lambda_h$. Here also note that the factor $(-a)$ does not depend on the height or the label of the step.
    \end{enumerate}
    For all cases, the weight of each step is recorded to the statistic, which validate our claim \eqref{eqn:wtblambda(w)}.
\end{proof}

\subsection{Proof of Proposition~\ref{Prop_spectral moments combinatorics}}

In this subsection, we complete the proof of Theorem~\ref{Thm_spectral moments} by establishing Proposition~\ref{Prop_spectral moments combinatorics}. We begin by deriving the closed formula for
\begin{equation}
    \alpha(a,b,c):=\sum_{M{\in}{\mathrm{Mat}}_{a,b,c}}q^{\mathrm{stat}(M)}
\end{equation}
as follows. Decompose ${\mathrm{Mat}}_{a,b,c}$ into
\begin{equation*}
    {\mathrm{Mat}}_{a,b,c}
    ={\mathrm{Mat}}^{(1)}_{a,b,c}
    \cup\Big(\bigcup_{x=2}^{a}{\mathrm{Mat}}^{(x)}_{a,b,c}\Big)
    \cup{\mathrm{Mat}}^{(v)}_{a,b,c} , 
\end{equation*}
where ${\mathrm{Mat}}^{(x)}_{a,b,c}$ denotes the set of the generalised matchings whose first vertex is 
\begin{enumerate}
    \item isolated if $x=1$,
    \smallskip 
    \item an opener of an arc whose closer is $x$ if $2\leq x\leq a$, and   \smallskip 
    \item a vertical line if $x=v$.
\end{enumerate}
Under this decomposition, we deduce the recurrence relation of $\alpha$ as follows.
\begin{itemize}
    \item If the first vertex is an isolated vertex, it contributes $2b+c$ to $\mathrm{stat}(M)$, and deleting the first vertex gives the generalised matching in ${\mathrm{Mat}}_{a-1,b,c}$. Thus we have
    \begin{equation*}
        \sum_{M\in{\mathrm{Mat}}^{(1)}_{a,b,c}}q^{\mathrm{stat}(M)}=q^{2b+c}\alpha(a-1,b,c). 
    \end{equation*}
    \item If the first vertex is connected to the vertex $x$ $(2\leq x\leq a)$ with an arc, then this arc $(1,x)$ contribute $x-2$ to $\mathrm{stat}(M)$, and deleting the arc $(1,x)$ gives the generalised matching in ${\mathrm{Mat}}_{a-2,b-1,c}$. Thus we have
    \begin{equation*}
        \sum_{M\in{\mathrm{Mat}}^{(x)}_{a,b,c}}q^{\mathrm{stat}(M)}=q^{x-2}\alpha(a-2,b-1,c).
    \end{equation*}
    By summing this over $2\le x\le a$, we obtain 
    \begin{equation*}
        \sum_{M\in\bigcup_{x=2}^{a}{\mathrm{Mat}}_{a,b,c}}q^{\mathrm{stat}(M)}
        =\sum_{x=2}^{a}q^{x-2}\alpha(a-2,b-1,c)
        =[a-1]_{q}\alpha(a-2,b-1,c).
    \end{equation*}
    \item If the first vertex is a vertical line, it does not affect $\mathrm{stat}(M)$, and deleting the first vertex gives the generalised matching in ${\mathrm{Mat}}^{(x)}_{a-1,b,c-1}$. Thus we have
    \begin{equation*}
        \sum_{M\in{\mathrm{Mat}}^{(v)}_{a,b,c}}q^{\mathrm{stat}(M)}=\alpha(a-1,b,c-1).
    \end{equation*}
\end{itemize}
Combining all of the above, we have the four-term recurrence
\begin{equation}\label{eqn:fourtermrecurrence}
    \alpha(a,b,c)=q^{2b+c}\alpha(a-1,b,c)+[a-1]_{q}\alpha(a-2,b-1,c)+\alpha(a-1,b,c-1)
\end{equation}
and it gives the explicit formula for $\alpha(a,b,c)$. Recall the expression $\mathsf{H}(b,c)$ given in \eqref{def of H}.
\begin{lem}\label{lem:alpha(a,b,c)}
    For $a\geq 2b+c$, we have
    \begin{equation}\label{eqn:alpha}
        \alpha(a,b,c)
        =\qbinom{a}{2b+c}[2b+c-1]_{q}!!\mathsf{H}(b,c).
    \end{equation}
\end{lem}
\begin{proof}
    We aim to validate that the right-hand side of \eqref{eqn:alpha} satisfies the equivalent equation of \eqref{eqn:fourtermrecurrence}
    \begin{equation}\label{eqn:alphafourtermequiv}
        \alpha(a,b,c)-q^{2b+c}\alpha(a-1,b,c)=[a-1]_{q}\alpha(a-2,b-1,c)+\alpha(a-1,b,c-1).
    \end{equation}
    Since $$[a]_{q}-q^{2b+c}[a-2b-c]_{q}=[2b+c]_{q},$$ the left-hand side of \eqref{eqn:alphafourtermequiv} becomes
    \begin{align*}
        \alpha(a,b,c)-q^{2b+c}\alpha(a-1,b,c)
        &=\frac{[a-1]_{q}!}{[2b+c-1]_{q}![a-2b-c]_{q}!}[2b+c-1]_{q}!!\mathsf{H}(b,c).
    \end{align*} 
    On the other hand, considering the case $j_{c}\leq b-1$ and $j_{c}=b$ separately, we have the recurrence relation of $\mathsf{H}(b,c)$
    \begin{align*}
        \mathsf{H}(b,c)
        &=\sum_{0\leq j_{1}\leq j_{2}\leq\cdots\leq j_{c}\leq b}\prod_{k=1}^{c}\frac{[2j_{k}+k-2]_{q}!!}{[2j_{k}+k-1]_{q}!!}\\
        &=\sum_{0\leq j_{1}\leq j_{2}\leq\cdots\leq j_{c}\leq b-1}\prod_{k=1}^{c}\frac{[2j_{k}+k-2]_{q}!!}{[2j_{k}+k-1]_{q}!!}
        +\sum_{0\leq j_{1}\leq j_{2}\leq\cdots\leq j_{c}= b}\prod_{k=1}^{c}\frac{[2j_{k}+k-2]_{q}!!}{[2j_{k}+k-1]_{q}!!}\\
        &=\mathsf{H}(b-1,c)+\frac{[2b+c-2]_{q}!!}{[2b+c-1]_{q}!!}\sum_{0\leq j_{1}\leq j_{2}\leq\cdots\leq j_{c-1}\leq b}\prod_{k=1}^{c-1}\frac{[2j_{k}+k-2]_{q}!!}{[2j_{k}+k-1]_{q}!!}\\
        &=\mathsf{H}(b-1,c)+\frac{[2b+c-2]_{q}!!}{[2b+c-1]_{q}!!} \mathsf{H}(b,c-1).
    \end{align*}
    Substituting \eqref{eqn:alpha} into the right-hand side of \eqref{eqn:alphafourtermequiv} gives
    \begin{align*}
        &\quad [a-1]_{q}\alpha(a-2,b-1,c)+\alpha(a-1,b,c-1)\\
        &=\frac{[a-1]_{q}!}{[2b+c-1]_{q}![a-2b-c]_{q}!}[2b+c-1]_{q}!!\Big(\mathsf{H}(b-1,c)+\frac{[2b+c-2]_{q}!!}{[2b+c-1]_{q}!!}\mathsf{H}(b,c-1) \Big)\\
        &=\frac{[a-1]_{q}!}{[2b+c-1]_{q}![a-2b-c]_{q}!}[2b+c-1]_{q}!!\mathsf{H}(b,c).
    \end{align*}
    In the last equality, we used the recurrence relation of $\mathsf H(b,c)$. As the both sides of \eqref{eqn:alphafourtermequiv} coincide, this completes the proof of the lemma.
\end{proof}

Next, for $i\leq j$, let ${\mathrm{Mat}}_{a,b,c}^{>j}(i)$ be the set of all generlised matchings in ${\mathrm{Mat}}_{a,b,c}^{>j}$ where exactly $i$ vertices among first $j$ vertices are openers. Equivalently, a generalised matching $M\in{\mathrm{Mat}}_{a,b,c}$ is contained in ${\mathrm{Mat}}_{a,b,c}^{>j}(i)$ if and only if the $i$ vertices among the first $j$ vertices of $M$ are openers, and other $j-i$ are isolated.

We proceed to evaluate the sum
\begin{equation}
    \sum_{M\in{\mathrm{Mat}}_{p+i,k,p-2k}^{>i}(i)}q^{\mathrm{stat}(M)}
\end{equation}
by introducing additional $i$ arcs to generalised matchings in $M\in{\mathrm{Mat}}_{p-i,k-i,p-2k}$ with $i$ openers are located at the left end of $M$. 

For this purpose, fix $M\in{\mathrm{Mat}}_{p-i,k-i,p-2k}$ and introduce two new vertices $o_1$ and $c_1$: $o_1$ at the left end of $M$ and $c_1$ at any other position. Then we introduce an arc connecting $o_1$ and $c_1$, which gives a generalised matching $M_1$ in ${\mathrm{Mat}}_{p-i+2,k-i+1,p-2k}^{>1}(1)$ as a result. In this step, there are $p-i+1$ possible locations of $c_1$, and these choices contribute a factor
\begin{equation*}
    1+q+\cdots +q^{p-i}=[p-i+1]_{q}
\end{equation*}
in total. We continue by adding two vertices: $o_2$ at the left end of $M_1$ and $c_2$ somewhere to the right of $o_1$. This gives the generalised matching $M_2$ in ${\mathrm{Mat}}_{p-i+4,k-i+2,p-2k}^{>2}(2)$ whose first two vertices are $o_2$ and $o_1$. Similarly to the case above, there are \( p - i + 2 \) possible positions for \( c_2 m \), and these choices contribute a factor 
\begin{equation*}
    q+q^2+\cdots+q^{p-i+2}=q[p-i+2]_{q}
\end{equation*}
in total. We repeat this procedure $i$ times to obtain a generalised matching $M_i$ in ${\mathrm{Mat}}_{p+i,k,p-2k}^{>i}(i)$. Combining these observations with Lemma \ref{lem:alpha(a,b,c)} gives that
\begin{align*}
    \sum_{M\in{\mathrm{Mat}}_{p+i,k,p-2k}^{>i}(i)}q^{\mathrm{stat}(M)}
    &=\prod_{m=1}^{i}q^{m-1}[p-i+m]_{q}\sum_{M\in{\mathrm{Mat}}_{p-i,k-i,p-2k}}q^{\mathrm{stat}(M)}\\
    &=q^{\frac{i(i-1)}{2}}\frac{[p]_{q}!}{[p-i]_{q}!}\qbinom{p-i}{p-2i}[p-2i-1]_{q}!!\mathsf{H}(k-i,p-2k).
\end{align*}
Finally, we consider a projection map
\begin{equation*}
    \mathcal{P}:{\mathrm{Mat}}_{p+j,k,p-2k}^{>j}(i)\rightarrow 
    {\mathrm{Mat}}_{p+i,k,p-2k}^{>i}(i)
\end{equation*}
defined by deleting first $j-i$ isolated vertices. Then, for a fixed generalised matching $M\in{\mathrm{Mat}}_{p+i,k,p-2k}^{>i}(i)$, there are $\binom{j}{i}$ choices to recover $\widetilde{M}\in P^{-1}(M)$ by introducing $j-i$ isolated vertices to $M$. Note that the contribution of newly added isolated vertex $v$ to $\mathrm{stat}(\widetilde{M})$ is $p+t$, where $t$ is the number of vertices between $v$ and $j+1$ which are not newly added. Then, the standard combinatorial interpretation of $q$-binomial coefficients gives that
\begin{equation*}
    \sum_{\tilde{M}\in\mathcal{P}^{-1}(M)}q^{\mathrm{stat}(\tilde{M})}
    =q^{\mathrm{stat}(M)+(j-i)(p-i)}\qbinom{j}{i}.
\end{equation*}
Therefore, we have
\begin{align*}
    \sum_{M\in{\mathrm{Mat}}_{p+j,k,p-2k}^{>j}(i)}q^{\mathrm{stat}(M)}
    &=q^{(j-i)(p-i)}\qbinom{j}{i}q^{\frac{i(i-1)}{2}}\frac{[p]_{q}!}{[p-i]_{q}!}\qbinom{p-i}{p-2i}[p-2i-1]_{q}!!\mathsf{H}(k-i,p-2k)\\
    &=q^{(j-i)(p-i)+\frac{i(i-1)}{2}}\qbinom{j}{i}\qbinom{p-i}{p-2i}\frac{[p]_{q}![p-2i-1]_{q}!!}{[p-i]_{q}!}\mathsf{H}(k-i,p-2k),
\end{align*}
which completes the proof.

\bigskip

\section{Analysis of scaled moments and densities} \label{Section_largeN asymptotic}

In this section, we perform asymptotic analysis and prove Theorems~\ref{Thm_genus expansion} and ~\ref{Thm_limiting density}.

\subsection{Proof of Theorem~\ref{Thm_genus expansion}}

We present the proof for the case where $p$ is even; the argument for odd $p$ proceeds analogously, with the substitution of $2p$ by $2p+1$. 

Notice first that by Theorem~\ref{Thm_spectral moments}, we have 
\begin{equation} \label{mN2p in terms of sum of FN}
   q^{p} m_{N,2p}^{(a)}=\sum_{l=0}^{p}F_{N}(l), \qquad F_N(l):=q^{p-l(2p-l)+\frac{l(l-1)}{2}}\sum_{k=0}^{p}\mathfrak{F}_{l}(k)\sum_{j=0}^{N-1}q^{j(2p-l)}\qbinom{j}{l},
\end{equation}
where 
\begin{equation} \label{def of mathfrak Fl(k)}
\mathfrak{F}_{l}(k)   :=(a+1)^{2p-2k}(-a)^{k}(1-q)^{k}\frac{[2p]_{q}!}{[2p-2l]_{q}!![l]_{q}!}\mathsf{H}(k-l,2p-2k). 
\end{equation} 
We now analyse the large-$N$ behaviour of the right-hand side of \eqref{mN2p in terms of sum of FN}. To this end, we reuse parts of the asymptotic analysis in \cite{BFO24}, which deals with the special case $a = -1$.

We first recall the definition of the Stirling numbers of the first kind \( s(n,k) \), which count the number of permutations of \( n \) elements with exactly \( k \) disjoint cycles. They can be expressed explicitly as 
\begin{equation} \label{def of Stirling number}
s(n,k) = (-1)^{n-k} \sum_{ 1 \le b_1 <  \dots < b_{n-k} \le n-1 } b_1 b_2 \dots b_{n-k}, \qquad (n > k \ge 1), 
\end{equation}
see \cite[Section~26.8]{NIST}.  
An alternative characterisation is via the generating function \cite[Eq.~(26.8.7)]{NIST}
\begin{equation}
\sum_{k=0}^n s(n,k)x^k = \frac{ \Gamma(x+1) }{ \Gamma(x-n+1) }.  
\end{equation} 
The following is given in \cite[Lemma 4.3]{BFO24}.

\begin{lem}
\label{lem:sumj}
Let $q$ be scaled according to \eqref{def of q scaling}. Then as $N\rightarrow\infty$, we have   
\begin{equation}
    \sum_{j=0}^{N-1}q^{j(2p-l)}\qbinom{j}{l}
    =\mathcal{C}_{l.0}\Big(\frac{N}{\lambda}\Big)^{l+1}+\mathcal{C}_{l,1}\Big(\frac{N}{\lambda}\Big)^{l}+\mathcal{C}_{l,2}\Big(\frac{N}{\lambda}\Big)^{l-1}+O(N^{l-2}).
\end{equation}
where
\begin{align*}
    \mathcal{C}_{l.0}
    &=\frac{(2p-l-1)!}{(2p)!}I_{1- \mathsf{s} }(l+1,2p-l),
    \\
    \mathcal{C}_{l,1}&=\frac{l(l+1)}{4}\frac{(2p-l-1)!}{(2p)!}I_{1-  \mathsf{s} }(l+1,2p-l)  -\frac{(l-1)}{2}\frac{(2p-l)!}{(2p)!}I_{1- \mathsf{s} }(l,2p-l+1)-\frac{1}{2}\frac{1}{l!} \mathsf{s}^{ 2p-l }(1- \mathsf{s} )^{l}, 
        \\
    \mathcal{C}_{l,2}&= \Big(-\frac{s(l+1,l-1)}{4}+\frac{l(1+l)(-4+l+9l^{2})}{144}\Big)\frac{(2p-l-1)!}{(2p)!}I_{1-\mathsf{s} }(l+1,2p-l) 
        \\
        &\quad -\frac{(l-1)(-2+7l+3l^{2})}{24}\frac{(2p-l)!}{(2p)!}I_{1-\mathsf{s} }(l,2p-l+1)
         \\
        &\quad + \frac{s(l,l-2) }{l(l-1)}\frac{(2p-l+1)!}{(2p)!}I_{1-\mathsf{s} }(l-1,2p-l+2) 
     +\frac{1}{12}\frac{1}{l!} \mathsf{s}^{ 2p-l }(1- \mathsf{s} )^{l}\Big(-\frac{3l^2+l+4p}{2}+(3l^2-2l)\frac{  \mathsf{s} }{1- \mathsf{s} }\Big).
\end{align*}
Here $\mathsf{s}=e^{-\lambda}.$
\end{lem}

Next, we examine the asymptotic behaviour of $\mathfrak{F}_{l}(k)$ as defined in \eqref{def of mathfrak Fl(k)}.  
Note that $\mathfrak{F}_{l}(k)$ takes nonzero values only for $l \leq k \leq p$, due to the term $\mathsf{H}(k - l, 2p - 2k)$.  
We begin by deriving the following asymptotic expansion.

\begin{lem}\label{lem:mfrakF1}
Let $q$ be scaled according to \eqref{def of q scaling}. Let $l \le k \le p$. Then as $N\rightarrow\infty$, we have
\begin{align}
\begin{split}
&\quad (1-q)^{k}\frac{[2p]_{q}!}{[2p-2l]_{q}!![l]_{q}!}
\\
&=\Big(\frac{\lambda}{N}\Big)^{k}
    \frac{(2p)!}{(2p-2l)!!l!}
    \Big(1+\frac{3l^2-4lp-2p^2-l+2p-2k}{4}\frac{\lambda}{N}
    +\mathfrak{a}'(p,l,k)\frac{\lambda^2}{N^2}
    +O(N^{-3})\Big), 
\end{split}
\end{align}
where 
\begin{equation*}
    \mathfrak{a}'(p,l,k)=\mathfrak{a}(p,l)+\frac{(p-k)(3p-3k-1+9l^2-12lp-6p^2-3l)}{24}
\end{equation*}
and 
\begin{align*}
    \mathfrak{a}(p,l) &=\frac{s(2p+1,2p-1)}{4}-s(p-l+1,p-l-1)-\frac{s(l+1,l-1)}{4}
    \\
    &\quad +\frac{1}{144}(12l-33l^2-30l^3+63l^4-8p+60lp+48l^2p-180l^3p+61p^2+126l^2p^2-4p^3-36p^4).
\end{align*} 
\end{lem}
\begin{proof}
The case $k = p$ was established in \cite[Lemma 4.1]{BFO24}. The result for general $l \le k \le p$ then follows directly from the case $k=p$ and  
\begin{equation*}
    (1-q)^{k}=\Big(\frac{\lambda}{N}\Big)^{k}+O(N^{-2k}). 
\end{equation*}
\end{proof}

Next, we proceed to find the asymptotic behaviour of
\begin{equation}\label{eqn:H(k-l,2p-l)}
    \mathsf{H}(k-l,2p-l)=\sum_{0\leq j_1\leq j_2\leq\cdots j_{2p-2k}\leq k-l}\prod_{r=1}^{2p-2k}\frac{[2j_{r}+r-2]_{q}!!}{[2j_{r}+r-1]_{q}!!}
\end{equation}
according to \eqref{def of mathfrak Fl(k)}. We first provide the asymptotic of the summand on the right-hand side of \eqref{eqn:H(k-l,2p-l)}.

\begin{lem}\label{lem:mfrakF2}
    Let $q$ be scaled according to \eqref{def of q scaling}. Let $l\leq k\leq p$ and $0\leq j_1 \leq j_2 \leq \cdots \leq j_{2p-2k}\leq k-l$ be given. Then as $N\rightarrow\infty$ we have
    \begin{equation}\label{Hsummandasymp}
        \prod_{r=1}^{2p-2k}\frac{[2j_{r}+r-2]_{q}!!}{[2j_{r}+r-1]_{q}!!}
        =\prod_{r=1}^{2p-2k}\frac{(2j_r +r-2)!!}{(2j_r+r-1)!!}
        \Big(1+\mathcal{D}_{k,1}\frac{\lambda}{N}+\mathcal{D}_{k,2}\frac{\lambda^2}{N^{2}}+O(N^{-3})\Big), 
    \end{equation}
    where
    \begin{align*}
        \mathcal{D}_{k,1}& =\frac{1}{2}\sum_{r=1}^{2p-2k}j_r + \frac{(p-k)(p-k+1)}{2}, 
        \\
        \mathcal{D}_{k,2}
        &= \sum_{1\leq r<s\leq p-k}\Big(\frac{j_{2r}+j_{2r-1}}{2}+r-1\Big) \Big(\frac{j_{2s}+j_{2s-1}}{2}+s-1\Big)
        \\
        &\quad+\sum_{r=1}^{p-k}\Big(\frac{(j_{2r}+r-1)(j_{2r-1}+r-1)}{4}
        +\frac{(j_{2r}+r-1)(j_{2r}+r-4)+(j_{2r-1}+r-1)(j_{2r-1}+r-2)}{24}\Big). 
    \end{align*}
\end{lem}

\begin{proof}
   We treat the cases where \( r \) is even and where \( r \) is odd separately.
 If $r=2r'$ is even, we have
    \begin{equation*}
        \frac{[2j_{r}+r-2]_{q}!!}{[2j_{r}+r-1]_{q}!!}
        =\frac{[2j_{2r'}+2r'-2]_{q}!!}{[2j_{2r'}+2r'-1]_{q}!!}
        =\prod_{j=1}^{j_{2r'}+r'-1}\frac{1-q^{2j}}{1-q^{2j+1}}.
    \end{equation*}
    Using the Talyor expansion
    \begin{equation*}
        1-q^{m}=1-e^{-\frac{m\lambda}{N}}=\frac{m\lambda}{N}\sum_{r=0}^{\infty}\frac{(-1)^{r}}{(r+1)!}\Big(\frac{m\lambda}{N}\Big)^{r},
    \end{equation*}
    and \eqref{def of Stirling number}, we have
    \begin{equation*}
        \prod_{j=1}^{n}(1-q^{2j})
        =\Big(\frac{\lambda}{N}\Big)^{n}(2n)!!\bigg(1-\frac{n(n+1)}{2}\frac{\lambda}{N}+\Big(s(n+1,n-1)+\frac{n(n+1)(2n+1)}{9}\Big)\frac{\lambda^{2}}{N^{2}}+O(N^{-3})\bigg)
    \end{equation*}
    as $N\rightarrow\infty$. Then, using  
    \begin{equation*}
        s(n+1,n-1)=\frac{n(n-1)(n+1)(3n+2)}{24}, \qquad (n\geq1)
    \end{equation*} 
    and taking the Talyor series of the reciprocal, we obtain
    \begin{equation}\label{eqn:r even}
        \prod_{j=1}^{n}\frac{1-q^{2j}}{1-q^{2j+1}}
        =\frac{(2n)!!}{(2n+1)!!}\bigg(1+\frac{n}{2}\frac{\lambda}{N}+\frac{n(n-3)}{24}\frac{\lambda^{2}}{N^{2}}+O(N^{-3})\bigg)
    \end{equation}
    as $N\rightarrow\infty$. Similarly, we have 
    \begin{equation}\label{eqn:r odd}
        \prod_{j=1}^{n}\frac{1-q^{2j-1}}{1-q^{2j}}
        =\frac{(2n-1)!!}{(2n)!!}\bigg(1+\frac{n}{2}\frac{\lambda}{N}+\frac{n(n-1)}{24}\frac{\lambda^{2}}{N^{2}}+O(N^{-3})\bigg). 
    \end{equation}
    Substituting \eqref{eqn:r even} and \eqref{eqn:r odd} into
    \begin{align*}
        \prod_{r=1}^{2p-2k}\frac{[2j_{r}+r-2]_{q}!!}{[2j_r+r-1]_{q}!!}
        &=\prod_{r=1}^{p-k}\frac{[2j_{2r-1}+2r-3]_{q}!!}{[2j_{2r-1}+2r-2]_{q}!!}\frac{[2j_{2r}+2r-2]_{q}!!}{[2j_{2r}+2r-1]_{q}!!}\\
        &=\prod_{r=1}^{p-k}\Big(\prod_{j=1}^{j_{2r-1}+r-1}\frac{1-q^{2j-1}}{1-q^{2j}}\prod_{j=1}^{j_{2r}+r-1}\frac{1-q^{2j}}{1-q^{2j+1}}\Big),
    \end{align*}
    we conclude the desired asymptotic behaviour \eqref{Hsummandasymp}.
\end{proof}

By combining Lemmas \ref{lem:mfrakF1} and \ref{lem:mfrakF2}, we find that \( \mathfrak{F}_{l}(k) \) is of order \( N^{-k} \) for \( k \geq l \) in the large-\( N \) limit. On the other hand, Lemma \ref{lem:sumj} already shows that \( \sum_{j=0}^{N-1} q^{j(2p - l)} \) is of order \( N^{l+1} \). Therefore, to derive the large-\( N \) expansion in \eqref{eqn:genus expansion} up to order \( N^{-1} \), it suffices to analyse the asymptotic behaviour of \( \mathfrak{F}_{l}(k) \) for \( k = l \), \( l+1 \), and \( l+2 \). While the full evaluation of the sum in \eqref{eqn:H(k-l,2p-l)} is challenging in general, the proof requires only three specific asymptotic expressions: \( \mathsf{H}(0, 2p - 2l) \), \( \mathsf{H}(1, 2p - 2l - 2) \), and \( \mathsf{H}(2, 2p - 2l - 4) \).

\begin{lem}\label{lem:H asymptotic}
    Let $q$ be scaled according to \eqref{def of q scaling}. 
Let $\mathsf{H}(b,c)$ be given by \eqref{def of H}.  Then as $N\rightarrow\infty$ we have
    \begin{align}
        \begin{split}
        \mathsf{H}(0,2p-2l)
        &=\frac{1}{(2p-2l-1)!!}\bigg(1+\frac{(p-l)(p-l-1)}{2}\frac{\lambda}{N}\\
        &\quad+\frac{(p-l)(p-l-1)(-4+13l+9l^2-13p-18lp+9p^2)}{72}\frac{\lambda^2}{N^2}+O(N^{-3})\bigg)\label{eqn:H(0,2p-2l)},
        \end{split}
        \\
        \mathsf{H}(1,2p-2l-2)
        &=\frac{p-l}{(2p-2l-3)!!}\bigg(1+\frac{(p-l-1)(3p-3l-4)}{6}\frac{\lambda}{N}+O(N^{-2})\bigg), 
    \label{eqn:H(1,2p-2l-2)}
    \\
    \label{eqn:H(2,2p-2l-4)}
        \mathsf{H}(2,2p-2l-4)& =\frac{(p-l)(p-l-1)}{2(2p-2l-5)!!}+O(N^{-1}).
    \end{align} 
\end{lem}
\begin{proof}
    Note that by \eqref{def of H special cases}, we have
    \begin{equation*}
        \mathsf{H}(0,2p-2l)=\frac{1}{[2p-2l-1]_{q}!!}.
    \end{equation*}
    Then \eqref{eqn:H(0,2p-2l)} follows from straightforward computations. 
    
    Next, we show  \eqref{eqn:H(1,2p-2l-2)}. 
    It follows from \eqref{Hsummandasymp} that
    \begin{align}
    \begin{split}
    \label{eqn:asymptotic of H(1,2p-2l-2)}
        \mathsf{H}(1,2p-2l-2)&=\sum_{0\leq j_{1}{\leq}\cdots{\leq}j_{2p-2l-2}\leq1}
        \prod_{r=1}^{2p-2l-2}\frac{(2j_{r}+r-2)!!}{(2j_{r}+r-1)!!}
        \\
        &\qquad \times \bigg(1+\Big(\frac{1}{2}\sum_{r=1}^{2p-2l-2}j_r +\frac{(p-l-2)(p-l-1)}{2}\Big)\frac{\lambda}{N}+O(N^{-2})\bigg).
    \end{split}
    \end{align} 
    Let $m$ be the number of indices $c$ such that $j_c=0$. Then we have
    \begin{align*}
        &\quad \sum_{0\leq j_{1}{\leq} \cdots{\leq}j_{2p-2l-2}\leq1}\prod_{r=1}^{2p-2k-2}\frac{(2j_{r}+r-2)!!}{(2j_{r}+r-1)!!}
        =\sum_{m=0}^{2p-2l-2}
        \Big(\prod_{r=1}^{m}\frac{(r-2)!!}{(r-1)!!}\Big)
        \Big(\prod_{r=m+1}^{2p-2l-2}\frac{r!!}{(r+1)!!}\Big)
        \\
        &=\sum_{m=0}^{2p-2l-2}
        \Big(\prod_{r=1}^{m}\frac{(r-2)!!}{(r-1)!!}\Big)
        \Big(\prod_{r=m+3}^{2p-2l}\frac{(r-2)!!}{(r-1)!!}\Big)= \sum_{m=0}^{2p-2l-2} (m+1)\prod_{r=1}^{2p-2l}\frac{(r-2)!!}{(r-1)!!} =  \sum_{m=0}^{2p-2l-2}  \frac{m+1}{(2p-2l-1)!!}. 
    \end{align*} 
    On the other hand, from the choice of the integer $m$, we also have
    \begin{equation*}
        \sum_{j=1}^{2p-2l-2}j_{r}=2p-2l-2-m. 
    \end{equation*}
    Combining all of the above, we obtain 
    \begin{equation*}
        \mathsf{H}(1,2p-2l-2)
        =\sum_{m=0}^{2p-2l-2}\frac{m+1}{(2p-2l-1)!!}\bigg(1+\Big(\frac{2p-2l-2-m}{2}+\frac{(p-l-2)(p-l-1)}{2}\Big)\frac{\lambda}{N}+O(N^{-2})\bigg),
    \end{equation*}
    which leads to \eqref{eqn:H(1,2p-2l-2)} after simplifications. 
    
   Finally, for \eqref{eqn:H(2,2p-2l-4)}, we only require the leading-order term. Thus, it suffices to establish the identity
    \begin{equation}
        \sum_{0\leq j_1\leq\cdots\leq j_{n}\leq 2}\prod_{r=1}^{n}\frac{(2j_r +r-2)!!}{(2j_r +r-1)!!}=\frac{n^4+10n^3+35n^2+50n+24}{8(n+3)!!}
    \end{equation}
    for $n\geq1$. This summation can be evaluated using a similar strategy to that used in the derivation of \eqref{eqn:H(1,2p-2l-2)}, by decomposing the sum according to the number of indices \( c \) for which \( j_c = 0 \), \( j_c = 1 \), or \( j_c = 2 \). This completes the proof. 
\end{proof}

We are now ready to show Theorem~\ref{Thm_genus expansion}. 

\begin{proof}[Proof of Theorem \ref{Thm_genus expansion}]
   By \eqref{mN2p in terms of sum of FN} and the observation following Lemma \ref{lem:mfrakF2}, the proof reduces to deriving the large $N$ asymptotic of
   \begin{equation*}
       \sum_{l=0}^{p}q^{p-l(2p-l)+\frac{l(l-1)}{2}}\Big(\mathfrak{F}_{l}(l)+\mathfrak{F}_{l}(l+1)+\mathfrak{F}_{l}(l+2) \Big)\sum_{j=0}^{N-1}q^{j(2p-l)}\qbinom{j}{l}.
   \end{equation*}
  By using Lemmas \ref{lem:mfrakF1} and \ref{lem:H asymptotic}, we obtain the following asymptotic expansions as $N \to \infty$: 
   \begin{align*}
       \mathfrak{F}_{l}(l)
       &=(a+1)^{2p-2l}\Big(-\frac{a\lambda}{N}\Big)^{l}\frac{(2p)!}{(2p-2l)!l!}\bigg(1+\frac{-1+5l^2-8lp}{4}\frac{\lambda}{N}\\
       &\quad +\frac{2l-21l^2-62l^3+225l^4+48lp+48l^2p-720l^3p+122p^2+84lp^2+576l^2p^2}{288}\frac{\lambda^2}{N^2}+O(N^{-3})\bigg), 
       \\
       \mathfrak{F}_{l}(l+1)&=(a+1)^{2p-2l+2} \Big(-\frac{a\lambda}{N}\Big)^{l+1}\frac{(2p)!}{(2p-2l-2)!l!} \Big(\frac{1}{2}+\frac{14+5l+15l^2-16p-30lp+6p^2}{24}\frac{\lambda}{N}+O(N^{-2})\Big), 
       \\
       \mathfrak{F}_{l}(l+2)&=(a+1)^{2p-2l+4} \Big(-\frac{a\lambda}{N}\Big)^{l+2}\frac{(2p)!}{(2p-2l-4)!l!}\Big(\frac{1}{8}+O(N^{-1})\Big),
   \end{align*}
   as $N \to \infty$. 
   Combining these expansions with Lemma \ref{lem:sumj}, the proof is completed after straightforward computation and simplification using the recurrence relation for the regularised incomplete beta function \cite[Eq.~(8.18.17)]{NIST}:
   \begin{equation}
       I_x(a,b)=I_x(a+1,b-1)+\frac{\Gamma(a+b)}{\Gamma(a+1)\Gamma(b)} x^a(1-x)^{b-1}.
   \end{equation} 
\end{proof}

\subsection{Proof of Theorem~\ref{Thm_limiting density}}
 
In this subsection, we show Theorem~\ref{Thm_limiting density}. 
We consider the Stieltjes transform
\begin{equation}
    G(y)\equiv G^{(a)}(y):=\int_{a}^{1} \frac{\rho^{(a)}(x)}{y-x} dx, \qquad y\in\C\backslash [a,1]
\end{equation}
of the limiting spectral density $\rho^{(a)}$. By definition, we have 
\begin{equation}
    G(y)=\sum_{p=0}^{\infty} \frac{1}{y^{p+1}}\int_{a}^{1} x^{p}{\rho^{(a)}}(x)dx=\sum_{p=0}^{\infty}\frac{1}{y^{p+1}}\mathcal{M}_{p,0}.
\end{equation}
Then it follows from \eqref{def of beta ftn} and \eqref{def of mathcal Mp0} that 
\begin{align*}
    G(y) &=\frac{1}{y}+\frac{1}{\lambda y}\sum_{p=1}^{\infty}\frac{1}{y^{p}}\sum_{l=0}^{\floor{p/2}}(a+1)^{p-2l}(-a)^{l}\binom{2l}{l}\binom{p}{2l}\int_{0}^{1-e^{-\lambda}}t^{l}(1-t)^{p-l-1}\, dt
    \\
    &=\frac{1}{y}+\frac{1}{\lambda y}
    \int_{0}^{1-e^{-\lambda}}\frac{1}{1-t}\sum_{l=0}^{\infty}\binom{2l}{l}\Big(\frac{-at(1-t)}{ y^{2}}\Big)^{l}\sum_{p=2l}^{\infty}\binom{p}{2l}\Big(\frac{(a+1)(1-t)}{y}\Big)^{p-2l}\, dt.
\end{align*}
For the inner summation, we apply the generalised binomial theorem:
\begin{equation*}
    \sum_{p=2l}^{\infty}\binom{p}{2l} z^{p-2l}= (1-z)^{-(2l+1)}, \qquad \abs{z}<1.
\end{equation*}
In our setting, the condition $| (a+1)(1-t)/y|<1$ is satisfied whenever
\begin{equation} \label{eqn:yrange0}
    \abs{y}> \abs{a+1}{e^{-\lambda}}.
\end{equation}
Therefore, for $y$ satisfying this range, we obtain
\begin{align*}
  G(y)= =\frac{1}{y}+\frac{1}{\lambda y}
    \int_{0}^{1-e^{-\lambda}}\frac{1}{1-t}\sum_{l=0}^{\infty}\binom{2l}{l}\Big(\frac{-at(1-t)}{ y^{2}}\Big)^{l} \Big(1-\frac{(a+1)(1-t)}{y}\Big)^{-(2l+1)}.  
\end{align*} 
Furthermore, by the generalised binomial theorem, we have the identity 
\begin{equation}
    \sum_{l=0}^{\infty}\binom{2l}{l}z^{l}=\frac{1}{\sqrt{1-4z}},\qquad \abs{z}<\frac14 . 
\end{equation}
In our context, this applies provided that 
\begin{equation} 
    \Big|{\frac{-at(1-t)}{(y-(a+1)(1-t)^{2}}}\Big|<\frac{1}{4}, \qquad t\in (0,1-e^{-\lambda})
\end{equation}
which is satisfied when $y$ lies in the range
\begin{equation}\label{eqn:yrange}
    y>
    \begin{cases}
        \mathsf u+\mathsf v &\lambda<\log(1-a), 
        \smallskip 
        \\    
        1&\lambda\geq\log(1-a), 
    \end{cases} \quad \textup{ and } \quad  y<
    \begin{cases}
        \mathsf u- \mathsf v&\lambda<\log(1-a)-\log(-a), 
          \smallskip 
        \\
        a&\lambda\geq \log(1-a)-\log(-a), 
    \end{cases} 
\end{equation} 
where $\mathsf{u}$ and $\mathsf v$ were defined in \eqref{def of u v(lambda) main}. 
Consequently, for all \( y \) in the range specified by \eqref{eqn:yrange0} and \eqref{eqn:yrange}, we obtain the simplified expression
\begin{equation}
    G(y) = \frac{1}{\lambda} \int_{0}^{1 - e^{-\lambda}} \frac{1}{1 - t} \cdot 
    \frac{1}{\sqrt{(y - (a+1)(1 - t))^2 + 4a t(1 - t)}} \, dt.
\end{equation}
 
To derive the density \( \rho^{(a)} \), fix \( x \in (a, 1) \) such that
\begin{equation}\label{eqn:x0x1 equation}
    (x-(a+1)(1-t))^{2}+4at(1-t)<0
\end{equation}
for
\begin{equation*}
    t\in(x_{0}-x_{1},\min\set{1-e^{-\lambda}, x_{0}+x_{1}}).
\end{equation*}  
Here, the constants $x_0$ and $x_1$ were defined in \eqref{eqn:def of x0x1}. Note that \( x_0 \pm x_1 \) are the roots of the quadratic expression on the left-hand side of \eqref{eqn:x0x1 equation}. Then, by the Stieltjes inversion formula, we obtain
\begin{equation}\label{eqn:integral rep of density}
\begin{split}
    {\rho}^{(a)}(x)
    &=\frac{1}{\pi}\lim_{\varepsilon\rightarrow0+} \im G(x-i\varepsilon)\\
    &=\frac{1}{\pi{\lambda}}\int_{x_{0}-x_{1}}^{\min\set{1-e^{-\lambda},x_{0}+x_{1}}}\frac{1}{1-t}\frac{1}{\sqrt{-(x-(a+1)(1-t))^2-4at(1-t)}} \, dt. 
\end{split}
\end{equation}
This integral can be simplified into a closed-form expression. To proceed, we first assume that $a\in[-1,0)$ and let
\begin{equation*}
    f(t):=\frac{1}{1-t}\frac{1}{\sqrt{-(x-(a+1)(1-t))^2-4at(1-t)}}\, dt. 
\end{equation*}
Then by \eqref{eqn:yrange}, it follows that
\begin{equation}\label{eqn:density}
    \begin{split}
        \rho^{(a)}(x)
        &=\frac{1}{\pi\lambda}\int_{x_{0}-x_{1}}^{1-e^{-\lambda}}f(t)dt\mathbbm{1}_{(\mathsf u- \mathsf v,\mathsf{u}+\mathsf{v})}\\
        &+\begin{cases}
            0 & \textup{if }\lambda \in (0, \log(1-a), \\
            \displaystyle\frac{1}{\pi\lambda} \int_{x_{0}-x_{1}}^{x_{0}+x_{1}}f(t)dt\mathbbm{1}_{(\mathsf u+ \mathsf v,1)} & \textup{if }\lambda \in ( \log(1-a), \log(1-a)-\log(-a)), \\
            \displaystyle\frac{1}{\pi\lambda}\int_{x_{0}-x_{1}}^{x_0+x_1}f(t)dt\mathbbm{1}_{(a,\mathsf u-\mathsf v)\cup(\mathsf u+ \mathsf v,1)} & \textup{if }\lambda \in ( \log(1-a)-\log(-a),\infty). 
        \end{cases}
    \end{split}
\end{equation} 
Note that $f(t)$ is of the form
\begin{equation*}
    f(t)=\frac{1}{1-a}\frac{1}{1-t}\frac{1}{\sqrt{(t-(x_0-x_1))(x_0+x_1-t)}}.
\end{equation*}
By using  
\begin{equation*}
    \int_{\alpha}^{\beta}\frac{1}{1-t}\frac{dt}{\sqrt{(t-\alpha)(\beta-t)}}=\pi\abs{(1-\alpha)(1-\beta)}^{-1/2},
\end{equation*}
we obtain
\begin{equation}\label{eqn:densityintegral1}
    \frac{1}{\pi\lambda}\int_{x_0-x_1}^{x_0+x_1}\frac{1}{1-t}\frac{dt }{\sqrt{(-(x-(a+1)(1-t))^2-4at(1-t)}}=\frac{1}{\lambda\abs{x}}.
\end{equation}
Similarly, by using
\begin{equation*}
    \int_{\alpha}^{\gamma}\frac{1}{1-t}\frac{dt}{\sqrt{(t-\alpha)(\beta-t)}}=\int_{\theta_{\gamma}}^{\pi}\frac{2 d\theta}{2-({\alpha+\beta})-({\beta-\alpha})\cos\theta},\qquad (\alpha<\gamma<\beta), 
\end{equation*}
where $\cos\theta_{\gamma}=(2\gamma-\alpha-\beta)/(\beta-\alpha)$, we obtain
\begin{equation}\label{eqn:densityintegral2}
    \frac{1}{\pi\lambda}\int_{x_0-x_1}^{1-e^{-\lambda}}\frac{1}{1-t}\frac{dt }{\sqrt{(-(x-(a+1)(1-t))^2-4at(1-t)}}
    = \frac{2}{\pi\lambda\abs{x}}\arctan{\sqrt{\frac{1-x_0-x_1}{1-x_0+x_1}\frac{1-e^{-\lambda}-x_0+x_1}{x_0+x_1-1+e^{-\lambda}}}}.
\end{equation}
Substituting \eqref{eqn:densityintegral1} and \eqref{eqn:densityintegral2} into \eqref{eqn:density} yields the explicit formula for the density function given in \eqref{def of limiting density}. Since the leading-order term $\rho^{(a)}$ is non-negative and integrates to one, it defines a probability density and can thus be integrated against any continuous test function.
For $a<-1$,  note that the symmetry \eqref{eqn:symmetry of spectral moment} implies a corresponding symmetry for the Stieltjes transform:
\begin{equation}
    G^{(1/a)}(y)=-\frac{1}{a}G^{(a)}(y).
\end{equation}
which in turn leads to the symmetry of the limiting density stated in \eqref{symmetry of rho a}. This completes the proof. \qed

\subsection{Proof of Proposition \ref{prop:Asymptotic zero distribution}}
In this subsection, we prove Proposition~\ref{prop:Asymptotic zero distribution} by exploiting the relation between the recurrence coefficients of the Al-Salam--Carlitz polynomials and the asymptotic distribution of their zeros. 

To apply the general result of Kuijlaars and Van Assche~\cite{KA99}, we work with the corresponding \emph{orthonormal} polynomials.
More precisely, in view of the orthogonality relation~\eqref{def of orthogonality AlSalam}, the orthonormal polynomials \( p_n^{(a)}(x) \equiv p_n^{(a)}(x;q) \) associated with the Al-Salam--Carlitz weight~\eqref{def of AlSalam weight} are given by
\begin{equation}
    p_{n}^{(a)}(x)=\Big((-a)^{n}(q;q)_{n}(1-q)q^{\frac{n(n-1)}{2}}\Big)^{-1/2}U_{n}^{(a)}(x).
\end{equation}
By~\eqref{def of three term AlSalam}, these orthonormal polynomials satisfy the three-term recurrence relation
\begin{equation} \label{eqn:orthonormal alsalam threeterm}
    x p_n^{(a)}(x) = a_{n+1} p_{n+1}^{(a)}(x) + b_n p_n^{(a)}(x) + a_n p_{n-1}^{(a)}(x),
\end{equation}
with recurrence coefficients
\begin{equation}
    a_n := \sqrt{ -a (1 - q^n) q^{n-1} }, \qquad b_n := (a + 1) q^n.
\end{equation}

With \( q \) scaled according to~\eqref{def of q scaling}, the recurrence coefficients admit the following limits:
\begin{equation}
  \lim_{n\rightarrow \infty} a_{n} \Big|_{q= e^{- \lambda/n }} = \sqrt{ -a e^{-\lambda} (1 - e^{-\lambda}) }, 
  \qquad 
  \lim_{n\rightarrow \infty} b_{n} \Big|_{q= e^{- \lambda/n }} = (a + 1) e^{-\lambda}.
\end{equation}
Define
\begin{equation} \label{def of alphabeta}
    \alpha(\lambda) := (a+1)e^{-\lambda} - 2 \sqrt{ -a e^{-\lambda} (1 - e^{-\lambda}) }, 
    \qquad 
    \beta(\lambda) := (a+1)e^{-\lambda} + 2 \sqrt{ -a e^{-\lambda} (1 - e^{-\lambda}) }.
\end{equation}
In terms of the parameters \( \mathsf{u} \) and \( \mathsf{v} \) introduced in~\eqref{def of u v(lambda) main}, we have 
\[
\alpha(\lambda) = \mathsf{u} - \mathsf{v}, \qquad \beta(\lambda) = \mathsf{u} + \mathsf{v}.
\]
Furthermore, we define a measure \( \omega_{[\alpha,\beta]} \) supported on the interval \( [\alpha, \beta] \) with density
\begin{equation}
    \frac{d\omega_{\alpha,\beta}}{dx} =
    \begin{cases}
        \displaystyle\frac{1}{\pi \sqrt{(\beta - x)(x - \alpha)}} & \text{if } x \in (\alpha, \beta),\smallskip \\
        0 & \text{otherwise.}
    \end{cases}
\end{equation}

Recall that \( \nu_n \) denotes the empirical zero distribution of the Al-Salam--Carlitz polynomials, as defined in~\eqref{def of ESD for AlSalam}. Then, applying the general result of~\cite[Theorem~1.4]{KA99}, we obtain in our setting
\begin{equation}
    \nu := \lim_{n \to \infty} \nu_n = \frac{1}{\lambda} \int_0^\lambda \omega_{[\alpha(s), \beta(s)]} \, ds.
\end{equation}
It follows that the support of the limiting zero distribution \( \nu \) is given by
\begin{equation}
    \Big[ \inf_{0 < s < \lambda} \alpha(s), \sup_{0 < s < \lambda} \beta(s) \Big],
\end{equation}
see also~\cite[Eq.~(1.10)]{KA99}. One can directly verify that this agrees with the expression in~\eqref{def of support}.
Moreover, the density function of the limiting zero distribution is given by
\begin{equation}
    x \mapsto \frac{1}{\pi \lambda} \int_{\lambda_{-}(x)}^{\min\{ \lambda, \lambda_{+}(x) \}} \frac{ds}{\sqrt{ (\beta(s) - x)(x - \alpha(s)) }},
\end{equation} 
where \( \lambda_{\pm}(x) \) denote the endpoints of the interval
\[
\{ s > 0 : \alpha(s) \leq x \leq \beta(s) \}.
\]
This expression agrees with the integral representation on the right-hand side of~\eqref{eqn:integral rep of density}, under the change of variables \( t = 1 - e^{-s} \), which recovers the limiting spectral density in~\eqref{def of limiting density}. \qed

\bibliographystyle{abbrv}

\end{document}